\pdfoutput=1
\documentclass[11pt,twoside]{article}

\usepackage{fancyhdr}
\usepackage[colorlinks,citecolor=blue,urlcolor=blue,linkcolor=blue,bookmarks=false]{hyperref}
\usepackage{amsfonts,epsfig,graphicx}
\usepackage{afterpage}
\usepackage{comment}
\usepackage{amsmath,amssymb,amsthm} 
\usepackage{fullpage}
\usepackage[T1]{fontenc} 
\usepackage{epsf} 
\usepackage{graphics} 
\usepackage{amsfonts,amsmath}
\usepackage[sort]{natbib} 
\usepackage{psfrag,xspace}
\usepackage{color,etoolbox}
\usepackage{subcaption}

\usepackage[utf8]{inputenc} % allow utf-8 input
\usepackage[T1]{fontenc}    % use 8-bit T1 fonts
\usepackage{url}            % simple URL typesetting
\usepackage{booktabs}       % professional-quality tables
\usepackage{amsfonts}       % blackboard math symbols
\usepackage{nicefrac}       % compact symbols for 1/2, etc.
\usepackage{microtype}      % microtypography
\usepackage{amsthm}
\usepackage{amsmath}
\usepackage{amssymb}
\usepackage{xcolor}
\usepackage{mathtools}
\usepackage{comment}
\usepackage[noend]{algorithmic}
\usepackage[Algorithm]{algorithm}
\usepackage{float}
\usepackage{xr}
\usepackage{xspace}
\usepackage{enumitem}
\usepackage[suppress]{color-edits}
\DeclareMathOperator*{\argmin}{arg\,min}

\newcommand{\var}{\mathrm{var}}
\newcommand{\logit}{\mathrm{logit}}

\newcommand{\E}{\mathbb{E}}

\newcommand{\p}{P}
\newcommand{\pn}{P_n}
\newcommand{\qn}{Q_n}

\newcommand{\pih}{\hat \pi}

\newcommand{\ind}{\mathbb{I}}

\newcommand{\Ind}[1]{\mathbb{I}\{#1\}}
\newcommand{\covar}{\mathrm{cov}}

\newcommand{\IF}{\mathrm{IF}}
\newcommand\norm[1]{\left\lVert#1\right\rVert}

\newcommand{\pmin}{\underline{\rho}}
\newcommand{\pmax}{\overline{\rho}}
\newcommand{\phih}{\hat{\phi}}

\newcommand{\pibar}{\bar{\pi}}

\newcommand{\muone}{\mu_1}
\newcommand{\muz}{\mu_0}

\theoremstyle{definition}
\newtheorem{theorem}{Theorem}[section]
\newtheorem{proposition}{Proposition}[section]
\newtheorem{corollary}{Corollary}[section]
\newtheorem{remark}{Remark}[section]
\newtheorem{assumption}{Assumption}[section]
\newtheorem{definition}{Definition}[section]

\newtheorem{lemma}{Lemma}[section]

\title{The role of the geometric mean in case-control studies
}
\author{
 Amanda Coston \\
  Heinz College \& Machine Learning Dept.\\
  Carnegie Mellon University\\
  \texttt{acoston@cs.cmu.edu}
  \and
  Edward H. Kennedy \\
  Dept. of Statistics  \& Data Science \\
  Carnegie Mellon University\\
 \texttt{edward@stat.cmu.edu} 
}
\begin{document}
\date{}
\maketitle

\begin{abstract}
% mobile applications?
% unemployment
% housing
Historically used in settings where the outcome is rare or data collection is expensive, outcome-dependent sampling is relevant to many modern settings where data is readily available for a biased sample of the target population, such as public administrative data. 
Under outcome-dependent sampling, common effect measures such as the average risk difference and the average risk ratio are not identified, but the conditional odds ratio is. Aggregation of the conditional odds ratio is challenging since summary measures are generally not identified. Furthermore, the marginal odds ratio can be larger (or smaller) than all conditional odds ratios. This so-called non-collapsibility of the odds ratio is avoidable if we use an alternative aggregation to the standard arithmetic mean. We provide a new definition of collapsibility that makes this choice of aggregation method explicit, and we demonstrate that the odds ratio is collapsible under geometric aggregation. We describe how to partially identify, estimate, and do inference on the geometric odds ratio under outcome-dependent sampling. Our proposed estimator is based on the efficient influence function and therefore has doubly robust-style properties. 
\end{abstract}

\section{Introduction}

Outcome-dependent sampling schemes like case-control sampling\footnote{In case-control sampling, `cases' are observations where the outcome occurred and `controls' are observations where the outcome did not occur. This is also known as case-referent sampling, case-noncase sampling, or choice-based sampling.} have long been used in medical, epidemiological, and survey research when the outcome is rare or the data are expensive to collect.  For example, pancreatic cancer is a particularly deadly but rare cancer whose risk factors have been analyzed under case-control designs \citep{hassan2007risk, lucenteforte2012alcohol}. Case-control methods are also often used in criminal justice settings to understand risk factors for crime and to audit for racial bias in police brutality \citep{loftin1988analysis, campbell2003risk, bogstrand2014drugs, wheeler2017factors, ridgeway2020role}.

A case-control design samples from two distributions: the population where the outcome occurred (‘cases’) and the population where the outcome did not occur (‘controls’) \citep{cornfield1951method,  miettinen1976estimability, breslow1996statistics}. An alternative outcome-dependent sampling scheme, known as case-population sampling, samples from the general population instead of the control distribution \citep{lancaster1996case, jun2020causal}. 
The broader class of biased sampling designs has a vast literature including choice-based sampling \citep{manski1977estimation, heckman2009note}, stratified sampling \citep{imbens1996efficient}, and exposure-based sampling such as the matched cohort design \citep{kennedy2015semiparametric}.

Under outcome-dependent sampling, common effect measures such as the average risk difference and the average risk ratio are not identified, but the conditional odds ratio is identified. 
%The conditional odds ratio is sufficient for optimal treatment assignment. 
The conditional odds ratio is often estimated under constant effect and/or parametric assumptions. Common in practice, logistic regression estimates the conditional odds ratio under the strong assumption that the parametric form is correctly specified \citep{prentice1979logistic}. %\textcolor{red}{[one could include interactions in a logistic regression model to relax the constant effect assumption, no?]}. 
A semiparametric approach relaxes these assumptions by placing restrictions directly on the odds ratio function, leaving other parts of the data-generating process unspecified \citep{chen2007semiparametric, tchetgen2010doubly}.
%However, in many settings the parametric assumption may be undesirable or even unrealistic. 
To avoid parametric assumptions altogether, a targeted maximum likelihood estimator (TMLE) with double robustness properties has been proposed for the population odds ratio, though only when the outcome rate is known \citep{van2008estimation}.

The odds ratio can approximate or bound other For example, the odds ratio approximates the risk ratio under a rare disease assumption  \citep{cornfield1951method, greenland1982need}. A recent work uses the odds ratio to upper bound the risk ratio under monotonicity assumptions on the treatment response and treatment selection \citep{jun2020causal}. To account for biased sampling, they target estimation of the outcome-conditional log odds ratio, and they propose a retrospective sieve logistic estimator to do so. %Our work also makes use of the outcome-conditional log odds ratio, but we propose a model agnostic estimation approach.
The odds ratio can also be used to define a measure of interaction effect when there is more than one exposure of interest \citep{vanderweele2011weighting}.

Our work considers odds ratios as the effect of interest, in a setting where the outcome rate is unknown. We compare several odds ratio-based measures of effect, discussing their usefulness (Section~\ref{section: problem setting}), identifiability (Section~\ref{section:identification}), and collapsibility (Section~\ref{section: collapsibility}). Of note is our novel observation that the geometric mean odds ratio is collapsible, unlike the standard (arithmetic aggregated) odds ratio. We offer a new definition of collapsibility that makes explicit the choice of aggregation. We provide the partial identification of the geometric mean odds ratio under outcome-dependent sampling. 
We propose a doubly robust-style estimation approach for the partially identified geometric odds ratio and describe conditions under which the estimator is $\sqrt{n}$-consistent and asymptotically normal in Section~\ref{section: estimation}. Lastly we describe how to do inference for the geometric odds ratio over a range of possible values for the unknown outcome rate in Section~\ref{section:inference}.
%achieves the efficiency bound. 
% 1. collapsibility
% 2. identifiability
% 3. doubly robust estimation
% 4. inference

\section{Problem setting} \label{section: problem setting}

Define $Z = (X, A, Y)$ where $X \in \mathbb{R}^d$ are covariates, $A \in \{0,1\}$ is a binary treatment or exposure, and $Y\in \{0,1\}$ a binary outcome. Let $\mathcal{P}$ denote the space of probability distributions. 

Our population of interest has distribution $P$. We observe data $(Z_1, Z_2,... Z_n)$ from a possibly biased distribution $Q$, where the bias results from outcome-dependent sampling. Let $\rho=P(Y=1)$ and $\omega=Q(Y=1)$ denote the probability that $Y=1$ under $P$ and $Q$ respectively.  %\textcolor{red}{(it might be good to use something other than $p$ and $q$ here, e.g., maybe $\delta$ or something - whenever I see $p/q$ I initially think of a density under $P/Q$)} 
The quantity 
$\omega$ is fixed and known--or if not known, can be easily estimated from the data. The quantity $\rho$ is fixed but generally unknown. We make virtually no assumptions about knowledge of $\rho$, requiring only that it lie in a user-specified range $[\pmin,~\pmax]$. We anticipate in many settings the user to generally have some knowledge about $\rho$ to inform the range, but if not, the user can specify $\pmin = \epsilon ,\pmax =1-\epsilon$.% absent any knowledge about $\rho$. 

While case-control data consists of samples from two different distributions, it can equivalently be viewed as independent and identically distributed draws from a modified distribution, known as Bernoulli sampling \citep{breslow2000semi}.
Under Bernoulli sampling, we first sample $Y \sim \mathrm{Bernoulli}(\omega)$, and then we sample $(X,A) \sim P(X, A \mid Y)$. We can relate $P$ and $Q$ as
\begin{align} \label{equation: bernoulli sampling}
    Q(z) = P(x,a \mid y)\omega = P(z)\frac{\omega}{\rho} %= p (y \mid x, a) p(x \mid a) p(a)\frac{q(y)}{p(y)} 
\end{align}

%This sampling approach can be considered a variant of the case-control design in which the number of cases and the number of controls are random numbers. %($\sum_{i=1}^n Y_i$)  ($\sum_{i=1}^n (1-Y_i)$). %, whereas typically case-control designs fix these at particular value(s) \citep{breslow2000semi}. 

We use potential (counterfactual) outcomes to define several measures of effect.
Let $Y^1$ denote the potential outcome that would have been observed under treatment and $Y^0$ the potential outcome that would have been observed under no treatment. 
Expectations are taken with respect to the target distribution $P$ unless otherwise noted via subscript.  

% \begin{definition} \label{defintion: ATE}
% The risk difference, also known as average treatment effect, describes the difference in risk of the outcome occurring under treatment and the risk under control.
% $$ RD := P(Y^1 =1) - P(Y^0 =1)$$
% \end{definition}

% \begin{definition}
% The risk ratio, also known as relative risk, describes the risk of the outcome occurring under treatment relative to the risk under control.
% $$ RR := \frac{P(Y^1 =1)}{P(Y^0 =1)}$$
% \end{definition}

\begin{definition}[Population odds ratio]
%\textcolor{red}{(i would maybe call this the marginal odds ratio here and in instances later)}\textcolor{purple}{Would this be confusing given our discussion of collapsibility in Section~\ref{section: collapsibility} where we marginal effect to refer to other aggregation methods.}
The population odds ratio describes the population odds of the outcome occurring under treatment relative to the odds under no treatment.
$$ OR := \frac{P(Y^1 =1)/P(Y^1=0)}{P(Y^0 =1)/P(Y^0=0)} = \frac{\mathrm{odds}(Y^1 =1)}{\mathrm{odds}(Y^0 =1)}$$
\end{definition}

The population odds ratio describes the ratio of the odds of the outcome in a world where everyone is treated to the odds of the outcome in a world where no one is treated. In some settings, such as when the outcome is very rare, we may be interested in interventions targeting certain strata of the population. Then our target of interest is the conditional odds ratio, as given in the following definition.

\begin{definition} [Conditional odds ratio]
The conditional odds ratio, also known as the adjusted odds ratio, describes the odds of the outcome occurring under treatment relative to the odds under no treatment for a particular stratum of the population with measured covariates $x$.
$$ OR(x) := \frac{P(Y^1 =1 \mid X = x)/P(Y^1=0\mid X = x)}{P(Y^0 =1 \mid X = x)/P(Y^0=0 \mid X = x)} = \frac{\mathrm{odds}(Y^1 =1 \mid X = x)}{\mathrm{odds}(Y^0 =1 \mid X = x)}$$
\end{definition}

Like many nonparametric function estimation problems, the conditional odds ratio is generally hard to estimate flexibly at fast rates. Logistic regression (without interactions) is probably the most commonly used method to estimate the conditional odds ratio in practice, but this approach will be biased when the conditional odds ratio varies across strata or when the linearity or logistic link function assumptions do not hold. On the other hand, nonparametric plugin methods suffer the curse of dimensionality. We may hope to avoid the curse of dimensionality by targeting the marginal effect since aggregated quantities can typically be estimated at faster rates even in nonparametric models. When we have randomly sampled iid observations, we can compare the importance of characterizing conditional effects against the reasonableness of the assumptions required, and we can then choose to target the marginal or conditional effect accordingly. When our data is from an outcome-dependent sample, however, this usual trade-off no longer applies because the marginal effect is not identified, as we show in the next section.

When the covariates are continuous or high-dimensional, we may desire an aggregation of the conditional odds ratio to describe a summary measure. The most common aggregation uses the arithmetic mean.

\begin{definition} [Arithmetic odds ratio]
The arithmetic aggregation of the conditional odds ratio is 
\begin{align}
    \alpha := \E[\mathrm{OR}(X)]
\end{align}

\end{definition}

A complicating issue is the population odds ratio could be smaller (or larger) than all conditional odds ratios. In other words, the odds ratio is not collapsible \citep{greenland1986identifiability, miettinen1981confounding, robinson1991some}. %However, perhaps surprisingly, aggregating odds ratios with the geometric mean yields a collapsible measure (see Section~\ref{section: collapsibility} for more on collapsibility); to the best of our knowledge this has not been noted before in the literature.  

\begin{definition}[Geometric odds ratio] \label{definition: geometric odds ratio}
The geometric odds ratio describes the geometric mean of the conditional odds ratios,
\begin{align} \label{equation: geometric odds ratio def1}
    %\gamma := \prod \Bigg\{ \frac{\frac{P(Y^1 =1 \mid X = x)}{P(Y^1=0\mid X = x)}}{\frac{P(Y^0 =1\mid X = x)}{P(Y^0=0\mid X = x)}}\Bigg\}^{dP(x)}
    \gamma := \prod \Bigg\{ OR(x)\Bigg\}^{dP(x)}
\end{align}
where $\prod$ denotes the product integral if $X$ is continuous or the product operator if $X$ is discrete. We can equivalently write the geometric odds ratio as 
\begin{align} 
    \label{equation: geometric odds ratio def3}
     %\gamma  &= \exp{ \Bigg( \E\Bigg[ \log \Bigg\{ \frac{\frac{P(Y^1 =1 \mid X)}{P(Y^1=0\mid X)}}{\frac{P(Y^0 =1\mid X)}{P(Y^0=0\mid X)}} \Bigg\} \Bigg]} \Bigg) = \exp{ \Big( \E\big[ \log ( OR(X) ) \big]} \Big)
     \gamma  &= \exp{ \Big( \E\big[ \log ( OR(X) ) \big]} \Big).
\end{align}
\end{definition}

\paragraph{Miscellaneous notation}
 $L \lesssim R$ indicates that $L \leq C \cdot R$ for some universal constant $C$. %what does universal mean? 
Define the squared $L_2(\p)$ norm of a function $f$ as ${\norm{f}^2 := \int (f(x))^2 dP(x)}$. $\Ind{}$ denotes the indicator function.
For samples $Z_1, Z_2,.... Z_n \sim Q$, we denote sample averages $\frac{1}{n}\sum_{i=1}^n f(Z_i)$ as $Q_n(f(Z))$. 
\section{Identifiability} \label{section:identification}

To identify our targets in terms of observable data, we make the following assumptions:

\begin{assumption}[Consistency]
\label{assumption:consistency}
A unit that receives treatment $a$ has outcome $Y = Y^a$.
\end{assumption}

\begin{assumption}[Ignorability]
\label{assumption:ignorability}
$(Y^0, Y^1) \perp A \mid X$. 
\end{assumption}
% this is stronger than actually required 

\begin{assumption}[Overlap]
\label{assumption:positivity}
$P(0 < P(A = 1 \mid X) < 1) = 1$. 
\end{assumption}

Assumption~\ref{assumption:consistency} requires that there is no interference between units; e.g., the potential outcome for a unit does not depend on the treatment assignment of other units. Assumption~\ref{assumption:ignorability} requires that the treatment assignment is as good as random conditional on measured covariates.
Assumption~\ref{assumption:positivity} requires that all covariate strata have some non-zero probability of receiving both treatment decisions.\footnote{For estimation we will require a stronger boundedness assumption on overlap.} In settings where it is unreasonable to make these assumptions, sensitivity analysis can be conducted in order to assess whether the results are sensitive to violations of the assumptions \citep{rosenbaum2010design, luedtke2015statistics, Robins(00)}.

%\textcolor{red}{(probably want to add a little discussion about the assumptions - also although they are pretty standard in stats, they are less common in econ where people often like IVs, so may want to hedge a little there or remove that word)}

\subsection{Point identification under random sampling}
If our data consists of random samples from the target distribution, then we can identify the population odds ratio, conditional odds ratio, arithmetic odds ratio, and geometric odds ratio respectively as

\begin{align} 
    OR &= \frac{\E[P(Y=1 \mid A = 1,X)]/\E[P(Y=0 \mid A = 1, X)]}{\E[P(Y=1 \mid A = 0, X)]/\E[P(Y=0 \mid A = 0, X)]} \\
    OR(x) &= %\frac{P(Y=1 \mid A = 1,X= x)/P(Y=0 \mid A = 1, X= x)}{P(Y=1 \mid A = 0, X= x)/P(Y=0 \mid A = 0, X= x)}  = 
    \frac{\mathrm{odds}(Y =1 \mid A = 1, X= x)}{\mathrm{odds}(Y=1 \mid A = 0, X =x)}\\
     \alpha &= \E \big[ OR(X) \big] \label{equation: identification of AOR under random sampling} \\
      \gamma &= \prod \Big\{OR(x)\Big\}^{dP(x)} \label{equation: identification of GOR under random sampling}
\end{align}

The conditional, marginal, and aggregated effects are all identified under random sampling. 
The next section will consider outcome-dependent sampling, where the conditional effect remains identified, but the marginal and aggregated effects are partially identified.
%Under random sampling, the conditional effect is identified but hard to estimate nonparametrically, whereas the marginal effect is also identified but easier to estimate, and may be of limited applicability for certain types of interventions \textcolor{red}{(might want to clarify what you mean here - also there are potentially many reasons why one might want a marginal or conditional effect)}. The next section will consider outcome-dependent sampling, where the conditional effect remains identified and harder to estimate, but the marginal effect is no longer identified. \textcolor{red}{(might want to clarify what easy/hard to estimate means - also some of this depends on some details, e.g., if the OR(X) is actually constant or very very smooth, it may be potentially as easy to estimate as alpha, even in a nonparametric model, and similarly if the nuisance functions are very hard to estimate then alpha will be as well - in other words there's some nuance here depending on exactly what is assumed, what the model is, etc. so may want to clarify a bit)} \textcolor{purple}{I am commenting this out for now since I think it may be confusing to have a discussion of estimation before the estimation section}

\subsection{Outcome-dependent sampling} \label{section: identification under outcome-dependent sampling}
When our samples $(Z_1, Z_2, ... Z_n)$ are drawn from the biased distribution $Q$, % (Eq.~\ref{equation: bernoulli sampling}), 
we can still identify the conditional odds ratio due to the symmetry %invariance?
of the odds ratio \citep{cornfield1951method}:  %\textcolor{red}{ this depends on the specifics of the biased distribution - this certainly isn't true for all biased distributions - do you mean for our particular biased distribution?)}

\begin{align*}
    OR(x) &= \frac{\mathrm{odds}(A=1 \mid Y = 1,X=x)}{\mathrm{odds}(A=1 \mid Y = 0, X=x)}
\end{align*}

We cannot point identify the arithmetic odds ratio and geometric odds ratio under outcome-dependent sampling because we sample from $P(X \mid Y)$, not $P(X)$. Without prior knowledge about the outcome rate $\rho$, we cannot estimate $P(X)$. We can partially identify the aggregation measures as a function of the unknown parameter $\rho$. 

We first define additional notation for regression functions:
\begin{align*}
    \mu_a(x) &= Q(Y =1 \mid X = x, A = a) \\
    \nu_a(x) &= P(Y =1 \mid X = x, A = a). 
\end{align*}
$\mu_a(x)$ is point identified under outcome-dependent sampling. We can partially identify $\nu_a(x)$ as a function of the unknown $\rho$ by applying Bayes' Rules to obtain

\begin{align}
    \nu_a(x; \rho) = \frac{P(A = a \mid X = x, Y = 1) P(X = x \mid Y = 1)\rho}{P(A = a, X = x \mid Y =1) \rho + P(A = a, X = x \mid Y =0) (1-\rho) }
\end{align}

The population odds ratio is partially identified as 
\begin{align}  \label{equation: identification of marginal odds ratio}
   \frac{\frac{\rho \E[ \nu_1(X) \mid Y = 1] + (1-\rho) \E[ \nu_1(X) \mid Y = 0]}{\rho \E[ 1- \nu_1(X) \mid Y = 1] + (1-\rho) \E[ 1-\nu_1(X) \mid Y = 0]}}{\frac{\rho \E[ \nu_0(X) \mid Y = 1] + (1-\rho) \E[ \nu_0(X) \mid Y = 0]}{\rho \E[ 1- \nu_0(X) \mid Y = 1] + (1-\rho) \E[ 1-\nu_0(X) \mid Y = 0]}}.
\end{align}
% \begin{align}
%     OR(\rho) &= \frac{\E \bigg[ \frac{P(A =1, X \mid Y = 1)}{\p(A=1,X)} \bigg] }{\E \bigg[ \frac{P(A =1, X \mid Y = 0)}{\p(A=1,X)} \bigg]} / \frac{\E \bigg[ \frac{P(A =0, X \mid Y = 1)}{\p(A=0,X)} \bigg]}{\E \bigg[ \frac{P(A =0, X \mid Y = 0)}{\p(A=0,X)} \bigg]} 
% \end{align}

% with $\p(A =a, X = x) = {\rho~\p(A= a, X = x \mid Y = 1) + (1-\rho)~\p(A= a, X = x \mid Y = 0)}$. \\

The partial identification of the arithmetic odds ratio is 
% quick find: definitions of unidentified target and discuss estimation strategy
\begin{align} \label{equation: target AOR in terms of py}
    \alpha(\rho)  = \rho \E[OR(X) \mid Y = 1]  + (1-\rho) \E[OR(X) \mid Y = 0].  
\end{align}

The geometric odds ratio is partially identified as 
\begin{align}\label{equation: target GOR in terms of py v1}
    \gamma(\rho)  = \Big( \prod_X OR(X)^{dP(X \mid Y = 1)} \Big)^{\rho}\Big( \prod_X OR(X)^{dP(X \mid Y = 0)} \Big)^{1-\rho} 
\end{align}
 
Or alternatively, 
\begin{align}\label{equation: target GOR in terms of py v2}
\begin{split}
    \gamma(\rho) &= \exp\Big( \rho~ \E[\log\big(OR(X)\big) \mid Y = 1] + (1-\rho) \E[\log\big(OR(X)\big) \mid Y = 0] \Big) \\
    &= \exp\Big( \rho~\E[\logit\big(\mu_1(X)\big)  - \logit\big(\mu_0(X)\big) \mid Y = 1] + (1-\rho)  \E[\logit\big(\mu_1(X)\big)  - \logit\big(\mu_0(X)\big) \mid Y = 0] \Big)
\end{split}
\end{align}

This identification indicates that $\E[\mathrm{logit}(\mu_a(X)) \mid Y = y]$ is a key object for the geometric odds ratio, and we will see later that this function plays a central role in the efficiency theory and estimation. We use the following to denote this function: 
 \begin{align} \label{equation: definition of y-conditional log gm target}
    \psi_{a,y} :=  \E[\mathrm{logit}(\mu_a(X)) \mid Y = y] 
 \end{align}

$\psi_{a,y}$ is identified under outcome-dependent sampling. To recap, Eqs.~\ref{equation: identification of marginal odds ratio}-\ref{equation: target GOR in terms of py v2} identify the marginal and aggregated odds ratios up to the unknown constant $\rho$.

\section{On collapsibility} \label{section: collapsibility}
  %   A measure is collapsible with respect to an aggregation method if the marginal measure can be expressed as a weighted aggregation  of the conditional measures using weights in \textbf{the probability simplex}.
  
  A desirable quality for a measure of effect is that the marginal effect describes the effect for a representative unit. The property of collapsibility (Def.~\ref{definition: collapsibility}) formalizes this quality \citep{whittemore1978collapsibility, greenland1999confounding}. 
 In this section, we take a detour from the biased sampling design to discuss collapsibility in detail. For simplicity our example will use random sampling, but the ideas are generally applicable.
  
  \begin{remark}
  In a departure from standard usage, we use the term       ``marginal'' to generically refer to a summary measure via an aggregation method that must be explicitly specified. As an example, the ``marginal odd'' of outcome $Y$ in standard usage unambiguously refers to $\frac{\p(Y=1)}{\p(Y=0)} = \frac{\E[\p(Y=1\mid X)]}{\E[\p(Y=0 \mid X)]}$. However, we denote this quantity as the marginal odds \emph{with respect to arithmetic aggregation}, differentiating it from other marginal measures such as, for example, the marginal odds with respect to geometric aggregation  $\frac{\prod \p(Y=1\mid X = x)^{dP(x)}}{\prod \p(Y=0\mid X = x)^{dP(x)}}$.
  \end{remark}

  Collapsibility is often discussed with respect to the arithmetic mean, under which collapsibility requires that we can specify weights for conditional effects such that the marginal effects equals their weighted average \citep{hernan2021causal}. For instance, the risk ratio, $RR = \E[Y^1]/\E[Y^0]$, is collapsible with respect to the arithmetic mean with weights $\frac{P(X)}{\E[Y^0]}\E[Y^0 \mid X]$.
The odds ratio, however, is not collapsible with respect to the arithmetic mean \citep{greenland1999confounding}. The population odds ratio is generally not equal to the conditional odds ratio, even if the conditional odds ratio is a constant. 

\paragraph{Example:} $X = \Ind{\mathrm{Female}}$ with $P(X =1) = 0.5$. 
% The risks under treatment and no treatment for women are $P(Y = 1 \mid X =1, A = 1) = \frac{1}{6}$ and $P(Y = 1 \mid X =1, A = 0) = \frac{9}{10}$. The risks for men are $P(Y = 1 \mid X =0, A = 1) = \frac{1}{26}$ and $P(Y = 1 \mid X =0, A = 0) = \frac{9}{14}$. 
The risks under treatment and no treatment for women and men are given in the below table along with the corresponding conditional risk ratios, $RR(X)$, and conditional odds ratios, $OR(X)$.
%The risk ratio for women is $\frac{5}{27}$ and for men $\frac{7}{117}$. 
The marginal risk ratio with respect to arithmetic aggregation is 
%$\frac{.5(1/6+1/26)}{.5(9/10+9/14)} =$ 
$\frac{140}{1053} \approx 0.133$. Averaging the conditional risk ratios with weights $\frac{P(X)}{\E[Y^0]}\E[Y^0 \mid X]$ yields the marginal risk ratio under arithmetic aggregation.
%$$\frac{0.5}{.5(9/10+9/14)}(9/10)*\frac{5}{27} +  \frac{0.5}{.5(9/10+9/14)}(9/14)*\frac{7}{117} = \frac{140}{1053}$$. 

    \begin{tabular}{c|c|c|c|c}
    X &    $P(Y = 1 \mid X, A = 1)$ &  $P(Y = 1 \mid X, A = 0)$ & RR(X) & OR(X) \\
    \hline 
      Female   & $\frac{1}{6}$ & $\frac{9}{10}$ &  $\frac{5}{27}$ & $\frac{1}{45}$\\[2mm]
      Male & $\frac{1}{26}$ & $\frac{9}{14}$ & $\frac{7}{117}$ & $\frac{1}{45}$
    \end{tabular}
    \vspace{3pt}\\
The conditional odds ratio for women equals the conditional odds ratio for men. %since  $\frac{\frac{1}{6}/\frac{5}{6}}{\frac{9}{10}/\frac{1}{10}} = \frac{\frac{1}{26}/\frac{25}{26}}{\frac{9}{14}/\frac{5}{14}}$ 
%$= \frac{1}{45}$.
However the marginal odds ratio under arithmetic aggregation is 
%$\frac{(\frac{1}{6} + \frac{1}{26})/(\frac{5}{6} + \frac{25}{26})}{(\frac{9}{10} + \frac{9}{14})/(\frac{1}{10} + \frac{5}{14})}$. 
$\approx \frac{3}{2} *\frac{1}{45}$.
It is not possible to find a weighted average of the conditionals that equals the marginal odds ratio.
 
While the odds ratio is not collapsible under the arithmetic mean, it is collapsible under the geometric mean. To demonstrate this, we introduce additional notation. Let ${f(a, b) \colon \mathbb{R}^2 \mapsto \mathbb{R}}$ denote an effect contrast and let $g_{w(x)}(P) \colon \mathcal{P} \mapsto \mathbb{R}$ denote a statistical functional that aggregates $X \sim P$ with weighting function $w(x)$. For example, letting $p(x)$ denote the density of random variable $X \sim P$, we describe the average risk difference (commonly referred to as average treatment effect) by specifying $g_{p(x)}(P) = \int x p(x) dx$
and $f(a,b) = a - b$.  We consider aggregations that can be written as a Fréchet mean--that is, there is an associated distance function $d$ such that $$g_{w(x)}(P) = \displaystyle \argmin_{z \in \mathcal{X}} \int_{\mathcal{X}} w(x) d^2(z, x) dx.$$
For ease of notation, we will write $g(X)$ to indicate $g(P)$ for the distribution $P$ over $X$.

  \begin{definition} \label{definition: collapsibility}
  A contrast $f$ is collapsible with respect to aggregation method $g$ if \begin{align}
     f\Big(g_{p(x)}(\mu_1(X)), g_{p(x)}(\mu_0(X))\Big) = g_{w(x)}\Big(f(\mu_1(X), \mu_0(X))\Big)
  \end{align}
  for weights $w(x)$ in the probability simplex and where $p(x)$ denotes the density or pmf of $X \sim P$. %\textcolor{red}{(mention what $p(x)$ is - i also think this could use some more detail on what $f$ and $g$ are, e.g., $f: \mathbb{R}^2 \mapsto \mathbb{R}$ whereas an equivalent statement for $g$ is a bit unclear)}
  \end{definition}
  
   The left hand side describes the marginal effect--that is, the contrast of the aggregations of $\mu_a(x)$ for $a \in \{0,1\}$. The right hand side describes a weighted aggregation of the conditional contrasts.
  
  Returning to our example, the average risk difference is collapsible with respect to the arithmetic mean using as weights the density of $x$. %This holds due to linearity of expectation.
%   \begin{align*}
%      f\Big(g(\mu_1(X)), g(\mu_0(X))\Big) &= 
%      g(\mu_1(X)) - g(\mu_0(X)) = g(\mu_1(X) - \mu_0(X)).
%   \end{align*}
  We briefly remark on the weights $p(x)$ and $w(x)$. While $p(x) = w(x)$ for the average risk difference, this need not be the case. The risk ratio has contrast $f(a,b) = \frac{a}{b}$ and is collapsible under aggregation $g_{w(x)}(P) = \int x w(x) dx$ with $w(x) = \frac{p(x)\E[Y^0 \mid X = x]}{\E[Y^0]}$. % this notation is slightly problematic.Should we define g(P(w(x))? but then it's unclear how we write definition 4.1. use a different RV on each side?

 As far as we are aware, this expanded definition of collapsibility that explicitly incorporates the aggregation method has not appeared in the literature before. With this machinery in place, we make the novel observation that the odds ratio is collapsible under geometric aggregation.

\begin{proposition}
The odds ratio is collapsible with respect to the geometric mean with weights $p(x)$.
\end{proposition}
% This result holds from
% Eq.~\ref{equation: geometric odds ratio def1} after applying the commutative property and power rule

% which can equivalently be written as 
% \begin{align*}
%     \label{equation: geometric odds ratio def2}
%     \gamma &= \frac{\frac{\prod P(Y^1 =1 \mid X = x)^{dP(x)}}{\prod P(Y^1=0\mid X = x)^{dP(x)}}}{\frac{\prod P(Y^0 =1\mid X = x)^{dP(x)}}{\prod P(Y^0=0\mid X = x)^{dP(x)}}}
% \end{align*}.
  %\textcolor{red}{(i think this should be highlighted more, e.g., in a separate proposition, and again we should mention that this hasn't been noted until now, as far as we know (if that's true)} 
  If the conditional OR is a constant $c$, then the geometric odds ratio also equals $c$. Recall that this was not necessarily the case for the arithmetic odds ratio, which could take on a value other than $c$. Since the geometric mean exhibits the desirable property of collapsibility, the remainder of this paper will consider estimation of the geometric odds ratio. We focus on the outcome-dependent sampling design, deferring results for the random sampling design to the Appendix. To motivate our estimation approach, we start by studying the efficiency theory.

\section{Efficiency} \label{section: efficiency}
\subsection{Preliminaries}
First, we restate and define additional notation for the nuisance functions:
\begin{align*}
    \mu_a(x) &:= Q(Y =1 \mid X =x, A =a) \quad \mathrm{for} \ \  a \in \{0,1\}\\
    \pi_a(x) &:= Q(A =a \mid X =x) \quad \mathrm{for} \ \ a \in \{0,1\} \\ 
    \eta(x) &:= Q(Y =1 \mid X =x) = \sum_{a=0}^1 \pi_a(x) \mu_a(x).
\end{align*}

Recall from Section~\ref{section: identification under outcome-dependent sampling} (Eq.s~\ref{equation: target GOR in terms of py v2}-\ref{equation: definition of y-conditional log gm target}) that we can write our target geometric OR as a function of $\rho$ and $\psi_{a,y}$ for $(a,y) \in \{0,1\}^2$:
 \begin{align} \label{equation: gamma as function of psi}
     \gamma(\rho) =  \exp\Big( \rho \big(\psi_{1,1} - \psi_{0,1} \big) + (1-\rho)\big(\psi_{1,0} - \psi_{0,0} \big) \Big)
 \end{align}
% We additionally define 
% \begin{align} \label{equation: definition of psi as log geometric mean}
%     \psi := \log(\gamma) %=  p\big(\psi_{1,1} - \psi_{0,1} \big) + (1-p)\big(\psi_{1,0} - \psi_{0,0} \big)
% \end{align}
    
We will first provide a von Mises-type expansion for $\psi_{a,y}$ using as an example $a=0$ and $y=1$ (see Appendix~\ref{appendix: general efficiency theory psi_ay} for the general result) and  subsequently provide the efficiency theory for $\gamma$. Functioning as a distributional analog to the Taylor expansion for real-valued functions, the von Mises-type expansion of a target parameter describes two key elements for efficiency theory: the influence function and a remainder term. Influence functions enable us to construct estimators with desirable properties, such as second-order bias, which can achieve fast convergence rates even in nonparametric settings. The remainder term plays an important role in characterizing the error of such estimators (see Section~\ref{section: estimation}). In a fully nonparametric model, the singular influence function is called the efficient influence function because it characterizes the efficiency bound in a local asymptotic minimax sense. The efficient influence function is therefore instructive for constructing optimal estimators. We direct the interested reader to \citet{bickel1993efficient, tsiatis2006semiparametric, kennedy2022semiparametric, hines2022demystifying} for more information on influence functions. %\textcolor{red}{(i would give the reader a heads up of why they should care about this expansion, some intuition for what it means, etc)}
We first define notation to refer to the nuisance functions on the distribution $\bar{Q}$: let $\bar{\eta}(x) := \bar{Q}(Y=1 \mid X=x)$ and similarly for $\bar{\pi}_0(x)$ and  $\bar{\mu}_0(x)$.

    %%%%%%%%%%%%%%%%%%%%%%%%% PSI efficiency %%%%%%%%%%%%%%%%%%%%%%%%
    \begin{lemma}
    \label{lemma: remainder term for psi_01}
    
    We have the following von Mises expansion for $\psi_{0,1}$:
    \begin{align*}
    \psi_{0,1}(Q) = \psi_{0,1}(\bar Q) + \int \varphi_{0,1}(\bar Q) d(Q-\bar Q) + R_2(\bar{Q}, Q&) \\
\mathrm{for} \quad     R_2(\bar{Q}, Q) =  \frac{\bar \omega - \omega}{\bar \omega} \big(\psi_{0,1}(Q) - \psi_{0,1}(\bar Q)\big) + \frac{1}{\bar \omega}\int  &\Big( \frac{\mu_0(x) - \bar \mu_0(x)}{\bar \mu_0(x)(1-\bar \mu_0(x))}\big( \eta(x) - \bar \eta(x) \big)  \\
        &+ \bar \eta(x) \frac{\mu_0(x) - \bar \mu_0(x)}{\bar \mu_0(x)(1-\bar \mu_0(x))} \frac{\bar \pi_0(x) - \pi_0(x)}{\bar \pi_0 (x)} \\
        &+ \eta(x)\frac{\mu^{*}_0(x) -1/2}{\mu^{*}_0(x)^2(1-\mu^{*}_0(x))^2}\big(\mu_0(x) - \bar \mu_0(x)\big)^2\Big)dQ
    \end{align*}
    
    where $\mu_0^{*}(x)$ lies between $\bar{\mu}_0(x)$ and $\mu_0(x)$  and 
    \begin{align*}
      \varphi_{0,1}(Z) &= \frac{\eta(X)\mathrm{logit}(\mu_0(X))}{\omega} - \psi_{0,1}  \\& +\frac{\eta(X)}{\omega} \frac{(1-A)(Y - \mu_0(X)}{\pi_0(X)\mu_0(X)(1-\mu_0(X))} + \frac{\mathrm{logit}(\mu_0(X))}{\omega}(Y- \eta(X)) \\
      &+ \psi_{0,1} - \frac{\psi_{0,1} Y}{\omega}.
    \end{align*}

    Since the remainder term $R_2(\bar{Q}, Q)$ is a second-order product of the nuisance function errors, we can apply Lemma 2 of \cite{kennedy2021semiparametric} to conclude that $\psi_{0,1}(Q)$ is pathwise differentiable with efficient influence function $\varphi_{0,1}(z; Q)$.
    \end{lemma}
    
    \begin{proof}
    Immediate from Lemma~\ref{lemma: remainder term for psi_ay} for $a = 0$ and $y=1$.
    \end{proof}

%%%%%%%%%%%%%%%%%%%% GAMMA efficiency %%%%%%%%%%%%%%%%%%%%%%%%%
\subsection{Efficiency theory for $\gamma$}
\begin{theorem} \label{theorem: remainder term for target}

We have the following von Mises-type expansion of our target $\gamma$:
%$\gamma = \exp(\psi(Q))$:
\begin{align*}
    \gamma(Q) &= \gamma(\bar{Q}) + \gamma(\bar{Q})\int \phi(\bar{Q})d(Q- \bar{Q})) + R_2(\bar{Q}, Q) \quad \mathrm{where} \\ 
        % & \phi(Q) =  \sum_{y=0}^1 \big(py + (1-p)(1-y)\big) \big( \varphi_{1,y}(Q) -  \varphi_{0,y}(Q)\big) \\
        & \phi(Q) =  \rho \big( \varphi_{1,1}(Q) -  \varphi_{0,1}(Q)\big) + (1-\rho)\big( \varphi_{1,0}(Q) -  \varphi_{0,0}(Q)\big).
\end{align*}

% \begin{align*}
%     \exp(\psi(Q)) &= \exp(\psi(\bar{Q})) + \exp(\psi(\bar{Q}))\int \phi(\bar{Q})d(Q- \bar{Q})) + R_2(\bar{Q}, Q) \quad \mathrm{where} \\ 
%         % & \phi(Q) =  \sum_{y=0}^1 \big(py + (1-p)(1-y)\big) \big( \varphi_{1,y}(Q) -  \varphi_{0,y}(Q)\big) \\
%         & \phi(Q) =  \rho \big( \varphi_{1,1}(Q) -  \varphi_{0,1}(Q)\big) + (1-\rho)\big( \varphi_{1,0}(Q) -  \varphi_{0,0}(Q)\big).
% \end{align*}

Then, by Lemma 2 of \citep{kennedy2021semiparametric}, our target $\gamma(Q)$ is pathwise differentiable with influence function $\gamma(Q)\phi(z;Q)$.
\end{theorem}

\begin{proof} \label{proof of corollary: remainder term for target}
% For a $\psi^{*}$ that lies between $\psi(\bar{Q})$ and $\psi(Q)$, applying Taylor's Theorem yields
% \begin{align} \label{equation: taylors on target}
%     \exp(\psi(Q)) &= \exp(\psi(\bar{Q})) + \exp(\psi(\bar{Q}))(\psi(Q) - \psi(\bar{Q}) + \frac{1}{2}\exp(\psi^{*}) (\psi(Q) - \psi(\bar{Q}))^2
% \end{align}

For a $\gamma^{*}$ such that $\log(\gamma^{*})$ lies between $\log(\gamma(\bar{Q}))$ and $\log(\gamma(Q))$, applying Taylor's Theorem yields
\begin{align} \label{equation: taylors on target}
    \gamma(Q) &= \gamma(\bar{Q}) + \gamma(\bar{Q})\Big(\log\big(\gamma(Q)\big) - \log\big(\gamma(\bar{Q})\big)\Big) + \frac{1}{2} \Big(\log\big(\gamma(Q)\big) - \log\big(\gamma(\bar{Q})\big)\Big)^2\gamma^{*}.
\end{align}

For the first order expression $\log(\gamma(Q)) - \log(\gamma(\bar{Q}))$ we apply Lemma~\ref{lemma: remainder term for psi_ay} to obtain
\begin{align*}
\log(\gamma(Q)) - \log(\gamma(\bar{Q})) %&= \sum_{y=0}^1 \rho_y(\psi_{1,y}(Q) - \psi_{1,y}(\bar Q) - (\psi_{0,y}(Q) - \psi_{0,y}(\bar Q))) \\
  &= \rho \Bigg(\int \varphi_{1,1}(\bar Q)d(Q- \bar Q) - \int \varphi_{0,1}(\bar Q)d(Q- \bar Q))\Bigg)
  \\&+ (1-\rho) \Bigg(\int \varphi_{1,0}(\bar Q)d(Q- \bar Q) - \int \varphi_{0,0}(\bar Q)d(Q- \bar Q))\Bigg) + R_2(\bar Q, Q)
\end{align*}
 where each term in the $R_2(\bar{Q}, Q)$ is a second-order nuisance function error. 

Substituting back into Eq.~\ref{equation: taylors on target} yields

% \begin{align*}
%     \exp(\psi(Q)) 
%     &= \exp(\psi(\bar{Q})) + \exp(\psi(\bar{Q}))\int \phi(\bar{Q})d(Q- \bar{Q})) + R_2(\bar{Q}, Q) \quad \mathrm{where} \\
%     & \phi(Q) =  \rho \big( \varphi_{1,1}(Q) -  \varphi_{0,1}(Q)\big) + (1-\rho) \big( \varphi_{1,0}(Q) -  \varphi_{0,0}(Q)\big)
% \end{align*}

\begin{align*}
    \gamma(Q) 
    &= \gamma(\bar{Q}) + \gamma(\bar{Q})\int \phi(\bar{Q})d(Q- \bar{Q})) + R_2(\bar{Q}, Q) \quad \mathrm{where} \\
    & \phi(Q) =  \rho \big( \varphi_{1,1}(Q) -  \varphi_{0,1}(Q)\big) + (1-\rho) \big( \varphi_{1,0}(Q) -  \varphi_{0,0}(Q)\big)
\end{align*}
\end{proof}

For quick reference we will restate the influence function of $\gamma$ using notation that defines $\eta_{y}(x) := Q(Y=y \mid x)$ and $\omega_y := Q(Y=y)$:

 \begin{align} \label{equation: influence function gamma}
    \mathrm{IF}(\gamma) &=  \Big(\rho \big( \varphi_{1,1}(Z) -\varphi_{0,1}(Z) \big) + (1-\rho) \big( \varphi_{1,0}(Z) -\varphi_{0,0}(Z) \big) \Big)\gamma 
 \end{align}
 where $\varphi_{a,y}(z;Q) = \mathrm{IF}(\psi_{a,y}(z;Q))$
 \begin{align} \label{equation: influence function of psi a y}
 \begin{split}
     = \frac{\mathrm{logit}(\mu_a(X))}{\omega_y}\Ind{Y=y} &- \psi_{a,y}  +\frac{\eta_{y}(X)}{\omega_y} \frac{\Ind{A= a}(Y - \mu_a(X))}{\pi_a(X)\mu_a(X)(1-\mu_a(X))}  \\
  &+ \psi_{a,y} - \frac{\psi_{a,y} \Ind{Y= y}}{\omega_y}
%      = \frac{\eta_y(X)\mathrm{logit}(\mu_a(X))}{\omega_y} &- \psi_{a,y}  +\frac{\eta_y(X)}{\omega_y} \frac{\Ind{A= a}(Y - \mu_a(X))}{\pi_a(X)\mu_a(X)(1-\mu_a(X))} + \frac{\mathrm{logit}(\mu_a(X))}{\omega_y}(\Ind{Y=y}- \eta_y(X)) \\
%   &+ \psi_{a,y} - \frac{\psi_{a,y} \Ind{Y= y}}{\omega_y}
  \end{split}
\end{align}

The influence function for $\gamma$ indicates that in addition to requiring that the propensity scores be bounded away from zero and one, we will additionally require that the conditional variances $(1-\mu_1(x))\mu_1(x)$ and $(1-\mu_0(x)) \mu_0(x)$ be bounded away from zero. This is notably different from the usual risk difference (ATE) setting, where we want the conditional variances to be small to improve efficiency.

\subsubsection*{Efficiency bound}
The efficiency bound describes the local asymptotic minimax lower bound on the mean squared error for any estimator of the target parameter, analogous to the Cramer-Rao bound for parametric settings. This bound provides a benchmark against which we can compare estimators. Additionally, the efficiency bound illuminates which components affect the difficulty of the estimation problem. For additional details we refer the reader to \citet{bickel1993efficient, van2003unified, tsiatis2006semiparametric, kennedy2022semiparametric}.
Before stating our efficiency bound, we introduce notation that will simplify the result. %\textcolor{red}{(similar to above, i think you need to remind the reader what the efficiency bound is, what it means, why they should care, cite the standard semiparametric refs (BKRW, van der Laan \& robins, tsiatis, etc)} 
We define the distribution-corrected log conditional odds ratio:
\begin{align} \label{equation: definition of big PSI}
    % \Psi(X,Y) &= \sum_{y=0}^1 \frac{\rho_y}{\omega_y}\Ind{Y=y} \Big( \logit(\mu_1(X))-  \logit(\mu_0(X))\Big) \\
     \Psi(X,Y) &=  \Big(Y\frac{\rho}{\omega} + (1-Y)\frac{1-\rho}{1-\omega} \Big)\Big( \logit(\mu_1(X))-  \logit(\mu_0(X))\Big) 
\end{align}

and the distribution-corrected aggregate log odds ratio:
\begin{align}\label{equation: definition of big PSI STAR}
    % \Psi^{*}(Y) &= \sum_{y=0}^1 \frac{\rho_y}{\omega_y} \Ind{Y=y} (\psi_{1,y} -\psi_{0,y}) \\
    \Psi^{*}(Y) &= Y\frac{\rho}{\omega} (\psi_{1,1} -\psi_{0,1}) + (1-Y)\frac{1-\rho}{1-\omega} (\psi_{1,0} -\psi_{0,0})
\end{align}

\begin{theorem} \label{theorem: efficiency bound for variance of influence function for geometric mean} The nonparametric efficiency bound for estimating $\gamma$ is given by $\sigma^2 := \var(\IF(\gamma)) = \gamma^2 \var(\IF(\log(\gamma)))$, where $\var(\IF(\log(\gamma)))$ equals 

\begin{align*}
    &\var(\Psi(X,Y)) + \var(\Psi^{*}(Y))-2~\mathrm{cov}\Big(\Psi(X,Y), \Psi^{*}(Y) \Big) + \\
 & \E \Bigg[ \Bigg( \frac{1}{\pi_1(X)\mu_1(X)(1-\mu_1(X))} +  \frac{1}{\pi_0(X)\mu_0(X)(1-\mu_0(X))} \Bigg)\Bigg( \Bigg(\frac{\rho -\omega}{\omega(1-\omega)} \Bigg)\eta(X) +\frac{1-\rho}{1-\omega} \Bigg)^2  \Bigg]
% &    \Bigg( \frac{p}{q} \Big( \psi_{1,1} - \psi_{0,1} \Big) - \frac{1-p}{1-q} \Big( \psi_{1,0} - \psi_{0,0} \Big)\Bigg)^2 q(1-q) - \\
\end{align*}
\end{theorem}

The proof is given in Appendix~\ref{appendix: proof of efficiency bound}.

The two terms that involve $\Psi^{*}(Y)$ result from having to estimate $\omega$; if $\omega$ is known by sampling design, then these terms drop from the bound. The coefficient $\big(\frac{\rho -\omega}{\omega(1-\omega)} \big)\eta(X) +\frac{1-\rho}{1-\omega} $ results from sampling bias. When $\rho=\omega$ this coefficient equals 1.\footnote{Compare to the bound under random sampling given in Appendix~\ref{section: efficiency under random sampling}}

Theorem~\ref{theorem: efficiency bound for variance of influence function for geometric mean} indicates that the difficulty of our estimation problem depends on the following factors:
\begin{enumerate}
    \item Heterogeneity in the distribution-corrected log conditional odds ratio $\Psi(X,Y)$;
    \item Heterogeneity in the distribution-corrected aggregate log odds ratio $\Psi^{*}(Y)$;
    \item Covariance in the distribution-corrected  conditional and aggregated log odds ratios;
    \item Propensity scores $\pi(x)$;
    \item Regression function variances $\muone(x)(1-\muone(x))$ and $\muz(x)(1-\muz(x))$; %(Equivalently conditional variances $\var(Y \mid X, A =a)$ for $a \in \{0,1\}$)
    \item Contrasts between the outcome rate in the target and sampled distributions, including the difference $\rho-\omega$ and the ratio $\frac{1-\rho}{1-\omega}$;
    \item Variance in the outcome in the sampled distribution, $\omega(1-\omega)$;
    \item Conditional outcome rates in the sampling distribution, $\eta(x)$.
\end{enumerate}
The second line of the efficiency bound shows the efficiency bound \emph{decreases} with the regression function variances, which is notably different from the average risk difference efficiency bound $\sigma^2_{ARD}$ \citep{hahn1998role}:
\begin{align*}
    \sigma^2_{ARD} = \E\Bigg[\frac{\muone(X)(1-\muone(X))}{\pi(X)} + \frac{\muz(X)(1-\muz(X))}{1-\pi(X)} + (\muone(X) - \muz(X) - \E[\muone(X) - \muz(X)])^2\Bigg]
\end{align*}

\section{Estimation} \label{section: estimation}
We propose a doubly robust style estimator for $\gamma$ that relies on estimation of $\log(\gamma)$. Before providing the error analysis of our proposed estimator, we briefly remark on the role of sample splitting. Estimating our nuisance functions on a separate sample, which we denote by $\hat Q$, that is independent of the sample denoted by $\qn$, enables us to avoid overfitting without having to rely on empirical process conditions.
With iid data, we can obtain these independent samples simply by randomly partitioning the data into two or more folds. More generally, one can use cross-fitting, a procedure which swaps the samples and averages the results, to regain sample efficiency  \citep{robins2008higher, zheng2010asymptotic, chernozhukov2018generic}. 
For simplicity we present our analysis under single sample splitting. We note that the outcome rate $\omega$ can be estimated on the full data sample.

\begin{theorem} \label{Theorem: Error of estimator}
Define the estimator for $\psi_{a,y}$ as
% \begin{align} \label{equation: proposed estimator for psi_ay}
% \begin{split}
% \hat \psi_{a,y} &:= \qn(\phi_{a,y}(Z; \hat{\mu}_a, \hat{\eta}_y, \hat{\pi}_a)) \quad \quad \mathrm{where} \\
% \phi_{a,y}(Z; \hat{\mu}_a, \hat{\eta}_y, \hat{\pi}_a) 
% &= \frac{\logit(\hat \mu_a(X))}{\hat \omega_y}\ind\{Y = y\}  +\frac{\hat \eta_y(X)}{\hat \omega_y } \frac{\ind\{A =a\}(Y -  \hat \mu_a(X))}{ \hat \pi_a(X) \hat \mu_a(X) \big(1- \hat \mu_a(X) \big)} 
% \end{split}
% \end{align}

\begin{align}\label{equation: proposed estimator for psi_ay}
\begin{split}
\hat \psi_{a,y} &:= \qn(\phi_{a,y}(Z; \hat{\mu}_a, \hat{\eta}, \hat{\pi}_a)) \quad \quad \mathrm{where} \\
\phi_{a,y}(Z; \hat{\mu}_a, \hat{\eta}, \hat{\pi}_a) 
&= \frac{\logit(\hat \mu_a(X))~\ind\{Y = y\}}{y\hat \omega +(1-y)(1-\hat \omega)} \\& +\Big(y\frac{\hat \eta(X)}{\hat \omega} + (1-y)\frac{1-\hat \eta(X)}{1-\hat \omega} \Big) \frac{\ind\{A =a\}(Y -  \hat \mu_a(X))}{ \hat \pi_a(X) \hat \mu_a(X) \big(1- \hat \mu_a(X) \big)} 
\end{split}
\end{align}

Assume the identification assumptions (\ref{assumption:consistency}-\ref{assumption:positivity}) hold 
and additionally assume the following five conditions hold. \\
\begin{enumerate}
    \item \emph{Convergence in probability in $L_2(\p)$ norm:} $\norm{\phi_{a,y} - \phih_{a,y}} = o_{\p}(1)$ . %This can be achieved by using plug-in estimates of the nuisance functions. %the consistency of the empirical process term
    \item \emph{Sample-splitting:} Nuisance functions $\hat{\pi}_1$, $\hat{\eta}$, $\hat{\mu}_1$, and $\hat{\mu}_0$ are estimated on $\hat Q$.\footnote{One could avoid the nuisance function estimation for $\eta$ since $\eta(x) = \pi_1(x)\mu_1(x) + \pi_0(x)\mu_0(x)$. However, in some cases it may be easier to estimate $\eta(x)$ directly than it is to estimate $\mu_a(x)$ or $\pi_a(x)$.}
    \end{enumerate}
And for some $\epsilon \in (0,1)$, 
\begin{enumerate}[resume]
    \item \emph{Strong overlap:} $Q(\epsilon < \pi_a(X) ) = 1$  and $Q(\epsilon < \pih_a(X)) = 1$ .
    \item \emph{Outcome variance:} $Q(\mu_a(X)(1-\mu_a(X)) > \epsilon) = 1$ and $Q(\hat{\mu}_a(X)(1-\hat{\mu}_a(X)) > \epsilon) = 1$ for $a \in \{0,1\}$. % f must be twice differentiable
    \item \emph{Outcome base rate:} $0 < \epsilon < \omega < 1-\epsilon$ and $0 < \epsilon < \hat \omega < 1-\epsilon$. 
    \end{enumerate}

Then the proposed estimator satisfies
\begin{align*}
\hat{\psi}_{a,y} - \psi_{a,y} &=  O_{\p}\Big(\left|\hat \omega - \omega \right|^2  + \norm{\hat \eta - \eta}\norm{\hat \mu_a - \mu_a} + \norm{\hat{\pi}_a - \pi_a}\norm{\hat{\mu}_a -\mu_a} + \norm{\hat{\mu}_a -\mu_a}^2 \Big)  \\
&+ (Q_n- Q) \Big(\phi_{a,y}(Z; Q) -\frac{\Ind{Y=y}\psi_{a,y}}{y\omega + (1-y)(1-\omega)}\Big) + o_{\p}\Big(\frac{1}{\sqrt{n}}\Big)
\end{align*}

\end{theorem}

%\textcolor{red}{(writing the estimator minus the target is more standard than what you've done here, which is the target minus the estimator)}

Theorem~\ref{Theorem: Error of estimator} demonstrates that our proposed estimator has second-order errors in the nuisance estimation errors, yielding ``doubly-fast'' rates. That is, we obtain a faster rate for our estimator even when estimating the nuisance function at slower rates. For example, to obtain $n^{-1/2}$ rates for our estimator, it is sufficient to estimate the nuisance functions at $n^{-1/4}$, allowing us to use flexible machine learning methods to nonparametrically estimate the nuisance functions under smoothness or sparsity assumptions.
%yields the desirable ``doubly-small" property \citep{leqi2021median}. \textcolor{red}{(i find it a bit odd  to only cite this paper here, we talked about this but i think many papers do (e.g., my density estimation paper with siva and larry, and probably lots of others) - i would add your own discussion of why second-order errors are useful)}
Since our error involves squared terms, this is not the usual double-robustness property that guarantees fast rates when \emph{either} of the propensity or regression function is estimated at fast rates. %; this result suggests we need to estimate \emph{both} $\mu_a$ and $\pi_a$ at fast enough rates.
%For example, to obtain $n^{-1/2}$ rates, it is not sufficient to estimate $\pi_a$ at rate $n^{-1/2}$; we must estimate $\mu_a$ at a fast enough rate (at least $n^{-1/4}$ rate). \textcolor{red}{(strictly speaking, the above result doesn't show classic double robustness is not possible, just that it doens't hold for the bound you have for this particular estimator. also i think in the last sentence above you are referring to getting root-n rates but this isn't clear) } 

%\subsubsection*{Proof of Theorem~\ref{Theorem: Error of estimator}}
\begin{proof}
Apply Theorem~\ref{Theorem: Error of generic function of regression functions} with $f(\mu_a(x)) =  \logit(\mu_a(x))$.
% \begin{align*}
% \hat{\psi} - \psi &= \hat{\eta}_1 - \eta_1 - (\hat{\eta}_0 - \eta_0) \\
% &= (\p- \pn) (\phi_1(Z; \p))+ O_{\p}\Big(\norm{\hat{\pi}_1 - \pi_1}\norm{\mu_1 -\hat{\mu}_1} + \norm{\mu_1 -\hat{\mu}_1}^2 \Big) + o_{\p}\Big(\frac{1}{\sqrt{n}}\Big) \\
% &+(\p- \pn) (\phi_0(Z; \p)) + O_{\p}\Big(\norm{\hat{\pi}_0 - \pi_0}\norm{\mu_0 -\hat{\mu}_0} + \norm{\mu_0 -\hat{\mu}_0}^2 \Big) + o_{\p}\Big(\frac{1}{\sqrt{n}}\Big) \\
% &=(\p- \pn) (\phi(Z; \p) + O_{\p}\Bigg(\sum_{a=0}^{1}\Big( \norm{\hat{\pi} - \pi}\norm{\mu_a -\hat{\mu}_a} + \norm{\mu_a -\hat{\mu}_a}^2 \Big) \Bigg) + o_{\p}\Big(\frac{1}{\sqrt{n}}\Big)
% \end{align*}

% Where the first equality holds by definition, the second uses Corollary~\ref{Corollary: Error of generic function of regression functions}, and the last line simplifies.
\end{proof}

% quick find
\begin{corollary}\label{corollary: asymptotic normality of log estimator}
The estimator $\hat{\psi}_{a,y}$ is $\sqrt{n}$-consistent and asymptotically normal 
under the assumptions in Theorem~\ref{Theorem: Error of estimator} and the following conditions:
\begin{enumerate}
    \item $\norm{\hat{\pi}_a -\pi_a} = O_{\p}(n^{-1/4})$
    \item $\norm{\hat{\mu}_a -\mu_a} = o_{\p}(n^{-1/4})$ for $a \in \{0,1\}$
     \item $\norm{\hat{\eta} -\eta} = O_{\p}(n^{-1/4})$
     \item $\left|\hat{\omega} -\omega\right| = o_{\p}(n^{-1/4})$ 
\end{enumerate}

The limiting distribution is $\sqrt{n}(\hat{\psi}_{a,y} - \psi_{a,y}) \rightsquigarrow \mathcal{N}\Big(0, \var\big(\IF(\psi_{a,y})\big)\Big)$ where $\var\big(\IF(\psi_{a,y})\big) = \var\Big(\phi_{a,y}(Z) - \frac{\Ind{Y= y}\psi_{a,y}}{\omega y + (1-\omega)(1-y)}\Big)$.% is given in Theorem~\ref{theorem: efficiency bound for variance of influence function for geometric mean}.
\end{corollary}

\subsection{Estimation of $\gamma$}
Our proposed estimator for $\gamma(\rho)$ is
\begin{align}
\begin{split}
    \hat{\gamma}(\rho) &=\exp\Big(  \rho~ (\hat{\psi}_{1,1} - \hat{\psi}_{0,1}) + (1-\rho) (\hat{\psi}_{1,0} - \hat{\psi}_{0,0})\Big) %\quad \mathrm{where} \\
%      \hat{\psi}_{a,y} &= \qn(\phi_{a,y}(Z; \hat{\mu}_a, \hat{\eta}_y, \hat{\pi}_a)) \quad \mathrm{where} \\
%     \phi_{a,y}(Z; \hat{\mu}_a, \hat{\eta}_y, \hat{\pi}_a) 
% &= \frac{\logit(\hat \mu_a(X))}{\hat \omega_y}\ind\{Y = y\}  +\frac{\hat \eta_y(X)}{\omega_y } \frac{\ind\{A =a\}(Y -  \hat \mu_a(X))}{ \hat \pi_a(X) \hat \mu_a(X) \big(1- \hat \mu_a(X) \big)} 
\end{split}
\end{align}
where $\hat{\psi}_{a,y}$ is defined in Eq.~\ref{equation: proposed estimator for psi_ay}.

\begin{corollary}\label{corollary: asymptotic normality of gamma estimator}
The estimator $\hat{\gamma}(\rho)$ is $\sqrt{n}$-consistent and asymptotically normal 
under the assumptions in Theorem~\ref{Theorem: Error of estimator} and in Corollary~\ref{corollary: asymptotic normality of log estimator} for all $(a,y) \in \{0,1\}^2$.

The limiting distribution is $\sqrt{n}(\hat{\gamma}(\rho) - \gamma(\rho)) \rightsquigarrow \mathcal{N}(0, \sigma^2)$ where $\sigma^2$ is given in Theorem~\ref{theorem: efficiency bound for variance of influence function for geometric mean}.
\end{corollary}

% estimation of other causal estimands
% [later] \subsection{Estimating a bound on the marginal odds ratio}
%We can estimate each of the four probabilities in doubly robust way, plug them in and use delta method. 

\section{Inference} \label{section:inference}
This section discusses how to do inference when estimating $\gamma(\rho)$ over a user-specified range of values $[\pmin,~\pmax]$ for the unknown outcome rate $\rho$.
We note that $\gamma(\rho)$ is monotonic in $\rho$, so our bound on $\gamma(\rho)$ has as endpoints $\gamma(\pmin)$ and $\gamma(\pmax)$. %, e.g. $[\min(\gamma(\pmin), \gamma(\pmax)), \max(\gamma(\pmin), \gamma(\pmax))]$
First we discuss how to obtain a confidence interval on the endpoints, using $\gamma(\pmax)$ as an example. Then we show how these imply a confidence interval on the bound for $\gamma(\rho)$. Guided by our theoretical results in the previous sections, our proposed approach is influence-function based. Alternatively, one could use the corrected confidence interval approach in \citet{imbens2004confidence} to give a confidence interval for $\gamma(\rho)$.

Based on the efficient influence function of $\gamma(\rho)$ (Eq.~\ref{equation: influence function gamma}), we create a pseudo-outcome $\zeta(Z_i; \hat \gamma(\pmax))$ for each observation $Z_i$ as 
 \begin{align}
    \zeta(Z_i; \hat \gamma(\pmax)) &= \hat \gamma(\pmax)  \Big(\pmax \big(\hat \varphi_{1,1}(Z_i) - \hat \varphi_{0,1}(Z_i) \big) + (1-\pmax) \big( \hat \varphi_{1,0}(Z_i) -\hat \varphi_{0,0}(Z_i) \big) \Big)
 \end{align}
 where $\hat \gamma(\pmax)$ and $\hat \psi_{a,y}$ for $(a,y) \in \{0,1\}^2$ are estimates using the doubly robust approach in Sec.~\ref{section: estimation} and where, for $a \in \{0,1\}$,
 \begin{align*} 
    \hat \varphi_{a,0}(Z) &= \frac{\mathrm{logit}(\hat \mu_a(X))}{1-\hat \omega}(1-Y)  +\frac{1-\hat \eta(X)}{1-\hat \omega} \frac{\Ind{A= a}(Y - \hat \mu_a(X))}{\hat \pi_a(X) \hat \mu_a(X)(1-\hat \mu_a(X))}  - \frac{\hat \psi_{a,y} (1-Y)}{1-\hat \omega} \\
    \hat \varphi_{a,1}(Z) &= \frac{\mathrm{logit}(\hat \mu_a(X))}{\hat \omega}Y  +\frac{\hat \eta(X)}{\hat \omega} \frac{\Ind{A= a}(Y - \hat \mu_a(X))}{\hat \pi_a(X) \hat \mu_a(X)(1-\hat \mu_a(X))}  - \frac{\hat \psi_{a,y} Y}{\hat \omega}. 
\end{align*}
\\
 %, and nuisance functions $\hat \pi(x)$, $\hat \eta(x)$, $\hat \mu_1(x)$ and $\hat \mu_0(x)$ are estimated flexibly estimated

We can then make use of the asymptotic normality results from the previous section. If the conditions in Corollary~\ref{corollary: asymptotic normality of gamma estimator} are met, then Corollary~\ref{corollary: asymptotic normality of gamma estimator} and Slutsky's theorem give that a $100(1-\alpha)\%$ asymptotic confidence interval for $\gamma(\pmax)$ is 
\begin{align*}
   \hat \gamma(\pmax) \pm z_{1-\alpha/2}\sqrt{\frac{\hat {\var}(\zeta(Z_i; \hat \gamma(\pmax)))}{n}}
\end{align*}
where $z_\beta$ is the standard normal quantile of $\beta$ and the empirical variance is over $Z$, holding as fixed the estimated nuisance functions. We can repeat this process on the same sample to obtain the confidence interval for $\gamma(\pmin)$.

To obtain the confidence interval for our bound on $\rho$, we define 
\begin{align*}
    \hat{\gamma}_{min}(\pmin,\pmax) &:= \min(\hat \gamma(\pmin),\hat \gamma(\pmax)) \\
    \hat{\gamma}_{max}(\pmin,\pmax) &:= \max(\hat \gamma(\pmin),\hat \gamma(\pmax)).
\end{align*}

\begin{proposition}
Under the conditions in Corollary~\ref{corollary: asymptotic normality of gamma estimator}, the interval $[l_{\alpha}, u_{\alpha}]$ where
\begin{align*} 
l_{\alpha} &:= \hat{\gamma}_{min}(\pmin,\pmax) - z_{1-\alpha/2}\sqrt{\frac{\hat \var(\zeta(Z_i; \hat{\gamma}_{min}(\pmin,\pmax)))}{n}}\\
u_{\alpha} &:=
\hat{\gamma}_{max}(\pmin,\pmax) + z_{1-\alpha/2}\sqrt{\frac{\hat \var(\zeta(Z_i; \hat{\gamma}_{max}(\pmin,\pmax))}{n}}
\end{align*}
\end{proposition}

gives a $100(1-\alpha)\%$ asymptotic confidence interval for the bound on $\gamma(\rho)$ for $p \in [\pmin, \pmax]$.

\begin{proof}
 $\hat{\gamma}_{min}$ and $\hat{\gamma}_{max}$ are asymptotically normal under Corollary~\ref{corollary: asymptotic normality of gamma estimator}. Applying Slutsky's Theorem, their variances are $\var(\zeta(Z_i; \hat{\gamma}_{min}(\pmin,\pmax)))$ and $\var(\zeta(Z_i; \hat{\gamma}_{max}(\pmin,\pmax))$, respectively. Applying Slutsky's Theorem once more, then $P(\rho >u_{\alpha}) = \frac{\alpha}{2}$ and $P(\rho < l_{\alpha}) = \frac{\alpha}{2}$. By the union bound, $P(\rho \notin [l_{\alpha},u_{\alpha}]) \leq \alpha$.
 
%   $\hat{\gamma}_{min}$ and $\hat{\gamma}_{max}$ are asymptotically normal under Corollary~\ref{corollary: asymptotic normality of gamma estimator}. Applying Slutsky's Theorem, their variances are $\var(\zeta(Z_i; \hat{\gamma}_{min}(\pmin,\pmax)))$ and $\var(\zeta(Z_i; \hat{\gamma}_{max}(\pmin,\pmax))$, respectively. Then $P(\rho >u_{\alpha}) = \frac{\alpha}{2}$ and $P(\rho < l_{\alpha}) = \frac{\alpha}{2}$. By the union bound, $P(\rho \notin [l_{\alpha},u_{\alpha}]) \leq \alpha$.
\end{proof}

\section{Conclusion}
The geometric mean odds ratio has the desirable property of collapsibility, unlike the more commonly used arithmetic mean. Under outcome-dependent sampling, the geometric odds ratio is not point identified, but we can estimate the geometric odds ratio as a function of the unknown outcome rate. We detail the efficiency theory for the geometric odds ratio, describe a doubly robust estimation procedure that is $\sqrt{n}$-consistent and asymptotically normal under mild conditions, and propose an inference procedure to construct confidence intervals for the geometric odds ratio over a range of possible values for the unknown outcome rate $\rho$.

\section*{Acknowledgments}
 Coston gratefully acknowledges financial support support from  the National Science Foundation  Graduate
Research Fellowship Program under Grant No. DGE1745016. Any opinions,
findings, and conclusions or recommendations expressed in this material are solely those of the authors.
 
% revision should add more to the conclusion 

\clearpage
\section{Appendix}
The appendix provides the derivations of identification results, proofs omitted from the main paper, and additional theoretical results.

\subsection{Identifications}
\subsubsection{Identifications under random sampling}

\begin{align*}
    OR &= \frac{\frac{\E[Y^1]}{\E[1-Y^1]}}{\frac{\E[Y^0]}{\E[1-Y^0]}} \\
    &= \frac{\frac{\E[ \E[Y^1 \mid X]]}{\E[\E[1-Y^1 \mid X]]}}{\frac{\E[\E[Y^0 \mid X]]}{\E[\E[1-Y^0 \mid X]]}} \\
    &= \frac{\frac{\E[ \E[Y^1 \mid X, A = 1]]}{\E[\E[1-Y^1 \mid X, A = 1]]}}{\frac{\E[\E[Y^0 \mid X, A = 0]]}{\E[\E[1-Y^0 \mid X, A = 0]]}} \\
    &= \frac{\frac{\E[ \E[Y \mid X, A = 1]]}{\E[\E[1-Y \mid X, A = 1]]}}{\frac{\E[\E[Y \mid X, A = 0]]}{\E[\E[1-Y \mid X, A = 0]]}} \\
\end{align*}
The first line holds by definition, the second by iterated expectation, the third from ignorability, and the fourth from consistency.

\begin{align*}
      \gamma &= \prod \Bigg\{ \frac{\frac{P(Y^1 =1 \mid X = x)}{P(Y^1=0\mid X = x)}}{\frac{P(Y^0 =1\mid X = x)}{P(Y^0=0\mid X = x)}}\Bigg\}^{dP(x)} \\
      &= \prod \Bigg\{ \frac{\frac{P(Y^1 =1 \mid A =1, X = x)}{P(Y^1=0 \mid A =1, X = x)}}{\frac{P(Y^0 =1 \mid A =0, X = x)}{P(Y^0=0 \mid A =0, X = x)}}\Bigg\}^{dP(x)} \\
      &= \prod \Bigg\{ \frac{\frac{P(Y =1 \mid A =1, X = x)}{P(Y=0 \mid A =1, X = x)}}{\frac{P(Y =1 \mid A =0, X = x)}{P(Y = 0 \mid A =0, X = x)}}\Bigg\}^{dP(x)} \\
\end{align*}
The first line holds by definition, the second from ignorability, and the third from consistency.

\subsubsection{Partial identifications under outcome-dependent sampling}

We first provide the derivation for the partial identification of $\nu_a(x) := P(Y =1\mid X = x, A = a)$

\begin{align*}
&= \frac{P(Y =1, X = x, A = a)}{P(X = x, A = a)} \\
 &= \frac{P(A = a \mid Y =1, X = x) P(X=x\mid Y =1) P(Y=1)}{P(X = x, A = a \mid Y = 1)P(Y=1) + P(X = x, A = a \mid Y = 0)P(Y=0)} \\&= 
    \frac{P(A = a \mid X = x, Y = 1) P(X = x \mid Y = 1)\rho}{P(A = a, X = x \mid Y =1) \rho + P(A = a, X = x \mid Y =0) (1-\rho) }
\end{align*}

Next we observe that under  consistency, positivity, and ignorability, we can express the population odds ratio in terms of $\nu_1(x)$ and $\nu_0(x)$ as 
\begin{align*}
    OR = \frac{\frac{\E[\nu_1(X)]}{\E[1-\nu_1(X)]}}{\frac{\E[\nu_0(X)]}{\E[1-\nu_0(X)]}}
\end{align*}
Then the law of total probability partially identifies the population odds ratio as 
\begin{align*}  
    OR(\rho) &= \frac{\frac{\rho \E[ \nu_1(X) \mid Y = 1] + (1-\rho) \E[ \nu_1(X) \mid Y = 0]}{\rho \E[ 1- \nu_1(X) \mid Y = 1] + (1-\rho) \E[ 1-\nu_1(X) \mid Y = 0]}}{\frac{\rho \E[ \nu_0(X) \mid Y = 1] + (1-\rho) \E[ \nu_0(X) \mid Y = 0]}{\rho \E[ 1- \nu_0(X) \mid Y = 1] + (1-\rho) \E[ 1-\nu_0(X) \mid Y = 0]}}
\end{align*}

\subsection{Proofs}
Some proofs will make use of the notation:
\begin{align*}
    \rho_y &:= P(Y= y)\ \  \mathrm{for} \ \ y \in \{0,1\} \\
    \omega_y &:= Q(Y= y)\ \  \mathrm{for} \ \ y \in \{0,1\} \\
    \eta_y(x) &:= Q(Y= y \mid X =x)\ \  \mathrm{for} \ \ y \in \{0,1\} \\
\end{align*}
The subscript will be omitted when context clearly indicates $y=1$ (following the notation used in the main paper).
Additionally we may overload some notation used in the main paper.

We first give a generic result for the von Mises expansion of smooth functions of the outcome regression functions that will be useful for deriving our main theoretical results.

\subsubsection{Second-Order Result for Functions of Regression Functions}
\begin{lemma}\label{lemma: second order result for functions of regression functions}
For $a \in \{0,1\}$,  $y \in \{0,1\}$, and any twice differentiable function $f$ of the regression function $\mu_a(x)$, define $\psi_{a,y} (Q) :=  \E[f(\mu_a(x)) \mid Y = y] = \int \frac{f(\mu_a(x))\eta_y(x)}{\omega_y}dQ $.

Then we can expand $\psi_{a,y}(Q)$ as

\begin{align*}
\psi_{a,y}(Q) &= \psi_{a,y}(\bar Q) + \int \varphi_{a,y}(\bar Q) d(Q-\bar Q) + R_2(\bar{Q}, Q)\\ 
 R_2(\bar{Q}, Q) &= \frac{\bar \omega_y - \omega_y}{\bar \omega_y} \big(\psi_{a,y}(Q) - \psi_{a,y}(\bar Q)\big) + \frac{1}{\bar \omega_y}\int  \Big( f{'}(\bar \mu_a(x))\big( \eta_y(x) - \bar \eta_y(x) \big) (\mu_a(x) - \bar \mu_a(x)) \\
    &+ \bar \eta_y(x) \big( f{'}(\bar \mu_a(x)) (\mu_a(x) - \bar \mu_a(x))\frac{\bar \pi_a(x) - \pi_a(x)}{\bar \pi_a (x)} + \eta(x)f{''}((\mu^{*}_a(x))\frac{(\mu_a(x) - \bar \mu_a(x))^2}{2}\Big)dQ
 \end{align*}

where $\mu_a^{*}(x)$ is between $\bar{\mu}_a(x)$ and $\mu_a(x)$ and where $\varphi_{a,y}$ is 
\begin{align} \label{equation: influence function for generic function of regression functions}
    % FULL VERSION WITH THE REDUNDANT TERM
    % \varphi_{a,y}(z;Q) &= \frac{\eta_y(X)f(\mu_a(X))}{q} - \frac{\psi_{a,y} \ind\{Y = y\}}{q} +\frac{f^{'}(\mu_a(X))\eta_y(X)}{q} \frac{\ind\{A =a\}(Y - \mu_a(X))}{\pi_a(X)} \\
    % &+ \frac{f(\mu_a(X))}{q}(\ind\{Y = y\}- \eta_y(X))
    \begin{split}
    \varphi_{a,y}(z;Q) &= \frac{f(\mu_a(X))}{\omega_y}\big(\ind\{Y = y\} -\eta_y(X)\big) +\frac{f^{'}(\mu_a(X))\eta_y(X)}{\omega_y} \frac{\ind\{A =a\}(Y - \mu_a(X))}{\pi_a(X)} \\
    &+ \frac{\eta_y(X) f(\mu_a(X))}{\omega_y}- \frac{\psi_{a,y} \ind\{Y = y\}}{\omega_y}
    \end{split}
\end{align}

Since the remainder term $R_2(\bar{Q}, Q)$ is a product of the nuisance function errors, we can apply Lemma 2 of \cite{kennedy2021semiparametric} to conclude that $\psi_{a,y}(Q)$ is pathwise differentiable with efficient influence function $\varphi_{a,y}(z; Q)$.
\end{lemma}

\subsubsection*{Proof of Lemma~\ref{lemma: second order result for functions of regression functions}}
\begin{proof}
We provide the proof for $a =1$ and $y=1$. Similar steps yield the result for other values of $(a,y) \in \{0,1\}^2$.

The posited influence function of $\psi_{1,1}$ 
\begin{align*}
    \varphi_{1,1}(Z) &= \frac{\eta(X)f(\muone(X))}{\omega_1} - \frac{\psi_{1,1} Y}{\omega_1} +\frac{f^{'}(\muone(X))\eta(X)}{\omega_1} \frac{A(Y - \muone(X))}{\pi_1(X)} + \frac{f(\mu_1(X))}{\omega_1}(Y- \eta(X))
\end{align*}
gives $\psi_{1,1}(Q) - \psi_{1,1}(\bar Q) - \int \varphi_{1,1}(\bar Q) d(Q-\bar Q) $
\begin{align*}
&= \psi_{1,1}(Q) - \psi_{1,1}(\bar Q)\\& \quad - \int \bigg( \frac{\bar \eta(X)f(\bar \muone(x))}{\bar \omega_1} - \frac{\psi_{1,1}(\bar Q) Y}{\bar \omega_1} +\frac{f^{'}(\bar \muone(x))\bar \eta(x)}{\bar \omega_1} \frac{A(Y - \bar \muone(x))}{\bar \pi(x)} + \frac{f(\bar \mu_1(x))}{\bar \omega_1}(Y- \bar \eta(x))\bigg)dQ\\
&= \frac{\bar \omega_1 - \omega_1}{\bar \omega_1} \big(\psi_{1,1}(Q) - \psi_{1,1}(\bar Q)\big) \\
& \quad - \frac{1}{\bar \omega_1}\int \Big( \bar \eta(x)f(\bar \muone(x))  - \eta(x)f(\muone(x)) + f^{'}(\bar \muone(x))\bar \eta(x) \frac{A(Y - \bar \muone(x))}{\bar \pi(x)} + f(\bar \mu_1(x))(Y- \bar \eta(x))\Big)dQ\\
&= \frac{\bar \omega_1 - \omega_1}{\bar \omega_1} \big(\psi_{1,1}(Q) - \psi_{1,1}(\bar Q)\big) \\
& \quad - \frac{1}{\bar \omega_1}\int \Big( \bar \eta(x)f(\bar \muone(x))  - \eta(x)f(\muone(x)) + f^{'}(\bar \muone(x))\bar \eta(x) \frac{\pi(x)(\muone(x) - \bar \muone(x))}{\bar \pi(x)} + f(\bar \mu_1(x))(\eta(x)- \bar \eta(x))\Big)dQ\\
&= \frac{\bar \omega_1 - \omega_1}{\bar \omega_1} \big(\psi_{1,1}(Q) - \psi_{1,1}(\bar Q)\big) \\
& \quad + \frac{1}{\bar \omega_1}\int \Big(  \eta(x) \big( f{'}(\bar \mu_1(x))(\muone(x) - \bar \mu_1(x)) + \eta(x)f{''}(\mu^{*}_1(x))\frac{(\muone(x) - \bar \mu_1(x))^2}{2} \\
& \quad - f^{'}(\bar \muone(x))\bar \eta(x) \frac{\pi(x)(\muone(x) - \bar \muone(x))}{\bar \pi(x)} \Big)dQ\\
&= \frac{\bar \omega_1 - \omega_1}{\bar \omega_1} \big(\psi_{1,1}(Q) - \psi_{1,1}(\bar Q)\big) \\
& \quad + \frac{1}{\bar \omega_1}\int \Big(   f{'}(\bar \mu_1(x)) \big( \eta(x) - \bar \eta(x) \big)(\muone(x) - \bar \mu_1(x))  \\&+ \bar \eta(x) \big( f{'}(\bar \mu_1(x)) (\muone(x) - \bar \mu_1(x))\frac{\bar \pi(x) - \pi(x)}{\bar \pi (x)} + \eta(x)f{''}(\mu^{*}_1(x))\frac{(\muone(x) - \bar \mu_1(x))^2}{2}  \Big)dQ\\
 \end{align*}
 
where the first equality makes use of the fact that $\int \varphi_{1,1}(\bar Q) d\bar Q = 0$, the second equality subtracts and adds term $\psi_{1,1}(Q)\omega_1/\bar \omega_1$ , and the third equality applies iterated expectation. 
The fourth equality applies a Taylor expansion for $f(\mu_1(x))$ around $f(\bar \mu_1(x))$ with the mean-value form remainder where
 $\mu^{*}_1(x)$ is between $\bar{\mu}_1(x)$ and $\mu_1(x)$. The final equality rearranges the first and third terms of the integrand by adding and subtracting $\bar \eta(x) \big( f{'}(\bar \mu_1(x)) (\muone(x) - \bar \mu_1(x))$. 
 
%  the expansion of the first and third term of the integrand is 
%  \begin{align*}
%      & \eta(x) \big( f{'}(\bar \mu_1(x))(\muone(x) - \bar \mu_1(x)) - f^{'}(\bar \muone(x))\bar \eta(x) \frac{\pi(x)(\muone(x) - \bar \muone(x)}{\bar \pi(x)} \\
%      & =  \eta(x) \big( f{'}(\bar \mu_1(x)) \bar \pi(x)\frac{(\muone(x) - \bar \mu_1(x))}{\bar \pi (x)} - f^{'}(\bar \muone(x))\bar \eta(x) \frac{\pi(x)(\muone(x) - \bar \muone(x)}{\bar \pi(x)} \\
%      & =  \eta(x) \big( f{'}(\bar \mu_1(x)) \bar \pi(x)\frac{(\muone(x) - \bar \mu_1(x))}{\bar \pi (x)} - \bar \eta(x) \big( f{'}(\bar \mu_1(x)) \bar \pi(x)\frac{(\muone(x) - \bar \mu_1(x))}{\bar \pi (x)} \\&+ \bar \eta(x) \big( f{'}(\bar \mu_1(x)) \bar \pi(x)\frac{(\muone(x) - \bar \mu_1(x))}{\bar \pi (x)} - f^{'}(\bar \muone(x))\bar \eta(x) \frac{\pi(x)(\muone(x) - \bar \muone(x)}{\bar \pi(x)} \\
%      & = \big( \eta(x) - \bar \eta(x) \big) f{'}(\bar \mu_1(x)) \bar \pi(x)\frac{(\muone(x) - \bar \mu_1(x))}{\bar \pi (x)}  \\&+ \bar \eta(x) \big( f{'}(\bar \mu_1(x)) (\muone(x) - \bar \mu_1(x))\frac{\bar \pi(x) - \pi(x)}{\bar \pi (x)} \\
%  \end{align*}
 
\end{proof}

As in the main paper, for our error analysis below we assume that our nuisance functions are estimated on a separate sample, which we denote by $\hat Q$, that is independent and of the same size as the sample denoted by $\qn$. With iid data, we can simply random split the data into two samples. To regain sample efficiency, we can use cross-fitting, a procedure which swaps the samples and averages the results \citep{gyorfi2006distribution, van2003unified, robins2008higher, zheng2010asymptotic, chernozhukov2018generic}. Sample splitting (or cross-fitting) enables us to avoid overfitting without having to rely on empirical process conditions.  

\begin{theorem} \label{Theorem: Error of generic function of regression functions}
Let $f$ be any twice differentiable function of the regression function $\mu_a$ with bounded second derivative. Define the estimator $\hat{\psi}_
{a,y}$ for the target $\psi_
{a,y} := \E[f(\mu_a(X)) \mid Y = y]$  as 
\begin{align} \label{equation: estimator for generic function of regression functions}
    \begin{split}
        \hat{\psi}_
{a,y} &:= Q_n (\phi_{a,y}(Z; \hat{\mu}_a, \hat{\eta}_y, \hat{\pi}_a)) \quad \mathrm{where} \\
&\phi_{a, y}(Z) = f(\mu_a(X))\frac{\eta_y(X)}{\omega_y} +f^{'}(\mu_a(X))\frac{\eta_y(X)}{\omega_y} \frac{\ind\{A = a\}(Y - \mu_a(X))}{\pi_a(X)} + \frac{f(\mu_a(X))}{\omega_y}(\ind\{Y = y\}-\eta_y(X))
    \end{split}
\end{align} Under the following conditions,
\begin{enumerate}
    \item $\norm{\phih_{a,y} - \phi_{a,y}} = o_{\p}(1)$ %\emph{Convergence in probability in $L_2(\p)$ norm:}
    \item \emph{Sample-splitting:} $\hat{\pi}_a$, $\hat{\eta}_y$, and $\hat{\mu}_a$ are estimated on samples from $\hat Q$.
    \item $Q(\pi_a(X) > \epsilon) = 1$ and $Q(\pih_a(X) > \epsilon) = 1$ for some $\epsilon > 0$ %\emph{Overlap:} 
    \item $\omega_y > \epsilon > 0$ and $\hat \omega_y > \epsilon > 0$ for $y \in \{0,1\}$ for some $\epsilon > 0$ 
    \item The target $\psi_{a,y}$ is bounded and well-defined, % so the target is not exploding or undefined
\end{enumerate}
then the estimator $\hat{\psi}_{a,y} :=   Q_n (\phi_{a,y}(Z; \hat{\mu}_a, \hat{\eta}_y, \hat{\pi}_a))$ satisfies 
\begin{align*}
 \hat{\psi}_{a,y} - \psi_{a,y} &=  O_{\p}\Big(\left|\hat \omega_y - \omega_y \right|^2  + \norm{\hat \eta_y - \eta_y}\norm{\hat \mu_a - \mu_a} + \norm{\hat{\pi}_a - \pi_a}\norm{\hat \mu_a -\mu_a} + \norm{\hat{\mu}_a -\mu_a}^2 \Big)  \\
&+ (Q_n - Q) \bigg(\phi_{a,y}(Z; Q) - \frac{\ind\{Y=y\}\psi_{a,y}(Q)}{\omega_y}\bigg) + o_{\p}\Big(\frac{1}{\sqrt{n}}\Big)
\end{align*}

\end{theorem}

\subsubsection*{Proof of Theorem~\ref{Theorem: Error of generic function of regression functions}}
\begin{proof}
\begin{align*}
\psi_{a,y}(Q) - \qn (\phi_{a,y}(Z; \hat{\mu}_a, \hat{\eta}_y, \hat{\pi}_a)  &= \overbrace{\psi_{a,y}(Q) - Q(\phi_a(Z; \hat{Q}))}^{A} + \overbrace{(Q- \qn) (\phi_a(Z; \hat{Q}) - \phi_a(Z; Q))}^B \\&+ \overbrace{(Q- \qn) (\phi_a(Z; Q))}^{C}
\end{align*}

For term A, we have that $\psi_{a,y}(Q) - Q(\phi_a(Z; \hat{Q}))$
\begin{align}
 %& P(\phi(Z;P) - \phi(Z;\hat P))
  &= \psi_{a,y}(Q) - \hat \psi_{a,y}(Q)  + \hat \psi_{a,y}(Q)\frac{\hat \omega_y - \omega_y}{\hat \omega_y} - \int \varphi_{a,z}(Z; \hat Q) dQ \\
  \begin{split} \label{equation: term A of general error result}
 &=  \frac{\hat \omega_y - \omega_y}{\hat \omega_y} \psi_{a,y}(Q) + \frac{1}{\hat \omega_y} \int  \Big( f{'}(\hat \mu_a(x))\big( \eta_y(x) - \hat \eta_y(x) \big) (\mu_a(x) - \hat \mu_a(x)) \\
    &+ \hat \eta_y(x) \big( f{'}(\hat \mu_a(x)) (\mu_a(x) - \hat \mu_a(x))\frac{\hat \pi_a(x) - \pi_a(x)}{\hat \pi_a (x)} + \eta(x)f{''}((\mu^{*}_a(x))\frac{(\mu_a(x) - \hat \mu_a(x))^2}{2}\Big)d
    Q
  \end{split}
    % &=  (\frac{\hat \omega_y - \omega_y}{\hat \omega_y} - \frac{\hat \omega_y - \omega_y}{\omega_y}) \psi_{a,y}(P) + \frac{\hat \omega_y - \omega_y}{\omega_y} \psi_{a,y}(P) + \frac{1}{\hat \omega_y} \int  \Big( f{'}(\hat \mu_a(x))\big( \eta_y(x) - \hat \eta_y(x) \big) (\mu_a(x) - \hat \mu_a(x)) \\
    % &+ \hat \eta_y(x) \big( f{'}(\hat \mu_a(x)) (\mu_a(x) - \hat \mu_a(x))\frac{\hat \pi_a(x) - \pi_a(x)}{\hat \pi_a (x)} + f{''}((\mu^{*}_a(x))\frac{(\mu_a(x) - \hat \mu_a(x))^2}{2}\Big)dP \\
    % &=   \frac{\hat \omega_y - \omega_y}{\omega_y} \psi_{a,y}(P)  - \frac{(\hat \omega_y - \omega_y)^2}{\omega_y\hat \omega_y} \psi_{a,y}(P) + \frac{1}{\hat \omega_y} \int  \Big( f{'}(\hat \mu_a(x))\big( \eta_y(x) - \hat \eta_y(x) \big) (\mu_a(x) - \hat \mu_a(x)) \\
    % &+ \hat \eta_y(x) \big( f{'}(\hat \mu_a(x)) (\mu_a(x) - \hat \mu_a(x))\frac{\hat \pi_a(x) - \pi_a(x)}{\hat \pi_a (x)} + f{''}((\mu^{*}_a(x))\frac{(\mu_a(x) - \hat \mu_a(x))^2}{2}\Big)dP \\
    % & = \frac{\hat \omega_y - \omega_y}{\omega_y} \psi_{a,y}(P) +  O_P \Big(\left|\hat \omega_y - \omega_y \right|^2  \\& \quad + \int (\hat \eta_y(x) - \hat \eta_y(x))(\hat \mu_a(x) - \mu_a(x)) + (\hat \mu_a(x) - \mu_a(x)))( \pi_a(x) - \hat \pi_a(x)) + (\hat \mu_a(x) - \mu_a(x))^2 dP \Big) \\
    % & = -(P-P_n)\bigg(\frac{\Ind{Y=y}\psi_{a,y}(P)}{\omega_y}\bigg)  \\
    % &+ O_P \Big( \left|\hat \omega_y - \omega_y \right|^2  + \norm{\hat \eta_y - \hat \eta_y}\norm{\hat \mu_a - \mu_a} + \norm{ \hat \mu_a - \mu_a}\norm{\hat \pi_a - \pi_a} + \norm{ \hat \mu_a - \mu_a}^2 \Big)
\end{align}
where the first equality holds by definition (Eq.~\ref{equation: estimator for generic function of regression functions} and~\ref{equation: influence function for generic function of regression functions}) and the second equality uses Lemma~\ref{lemma: second order result for functions of regression functions}.

We can write the first term as
\begin{align*}
 \frac{\hat \omega_y - \omega_y}{\hat \omega_y} \psi_{a,y}(Q)  &=  (\frac{\hat \omega_y - \omega_y}{\hat \omega_y} - \frac{\hat \omega_y - \omega_y}{\omega_y}) \psi_{a,y}(Q) + \frac{\hat \omega_y - \omega_y}{\omega_y} \psi_{a,y}(Q) \\
     &=   \frac{\hat \omega_y - \omega_y}{\omega_y} \psi_{a,y}(Q)  - \frac{(\hat \omega_y - \omega_y)^2}{\omega_y\hat \omega_y} \psi_{a,y}(Q)  \\
     & = -(Q-Q_n)\bigg(\frac{\Ind{Y=y}\psi_{a,y}(Q)}{\omega_y}\bigg)  + O_P \Big( \left|\hat \omega_y - \omega_y \right|^2 \Big)
\end{align*}
where the last line applies the conditions given in Theorem~\ref{Theorem: Error of generic function of regression functions}.

Using this result and again applying applies the conditions given in Theorem~\ref{Theorem: Error of generic function of regression functions} to our expression for Term A in Eq.~\ref{equation: term A of general error result}, we have that $\psi_{a,y}(Q) - Q(\phi_a(Z; \hat{Q}))$
\begin{align*}
 %& P(\phi(Z;P) - \phi(Z;\hat P))
  & = -(Q-Q_n)\bigg(\frac{\Ind{Y=y}\psi_{a,y}(Q)}{\omega_y}\bigg)  +  O_P \Big(\left|\hat \omega_y - \omega_y \right|^2  \\& \quad + \int (\eta_y(x) - \hat \eta_y(x))( \mu_a(x) - \hat \mu_a(x)) + (\hat \mu_a(x) - \mu_a(x)))( \pi_a(x) - \hat \pi_a(x)) + (\hat \mu_a(x) - \mu_a(x))^2 dQ \Big) \\
    & = -(Q-Q_n)\bigg(\frac{\Ind{Y=y}\psi_{a,y}(Q)}{\omega_y}\bigg)  \\
    &+ O_P \Big( \left|\hat \omega_y - \omega_y \right|^2  + \norm{\hat \eta_y - \eta_y}\norm{\hat \mu_a - \mu_a} + \norm{ \hat \mu_a - \mu_a}\norm{\hat \pi_a - \pi_a} + \norm{ \hat \mu_a - \mu_a}^2 \Big)
\end{align*}
where  the second equality uses Cauchy-Schwarz inequality.

For term B, since $Q_n$ is the empirical measure on an independent sample from $\hat{Q}$, we can apply Lemma 2 of \citet{kennedy2020sharp} with our assumption that $\norm{\phi_{a,y}(Z; \hat{Q}) - \phi_{a,y}(Z; Q)} = o_{\p}(1)$:
\begin{align*}
    (Q- \qn) (\phi_{a,y}(Z; \hat{Q}) - \phi_{a,y}(Z; Q)) = O_{\p}\Big(\frac{\norm{(\phi_{a,y}(Z; \hat{Q}) - \phi_{a,y}(Z; Q)}}{\sqrt{n}}\Big) = o_{\p}\Big(\frac{1}{\sqrt{n}}\Big)
\end{align*}

We combine the results for term A and term B with term C to obtain:
\begin{align*}
\psi_{a,y}(Q) - \qn (\phi_{a,y}(Z; \hat{\mu}_a, \hat{\eta}_y, \hat{\pi}_a))  &= (Q- \qn) \bigg(\phi_a(Z; Q)) -\frac{\Ind{Y=y}\psi_{a,y}(Q)}{\omega_y}\bigg)  + O_P \Bigg( \left|\hat \omega_y - \omega_y \right|^2  \\&+ \norm{\hat \eta_y - \hat \eta_y}\norm{\hat \mu_a - \mu_a} + \norm{ \hat \mu_a - \mu_a}\norm{\hat \pi_a - \pi_a} + \norm{ \hat \mu_a - \mu_a}^2 \Bigg) + o_{\p}\Big(\frac{1}{\sqrt{n}}\Big)
\end{align*}

\end{proof}

\subsection{General efficiency theory for $\psi_{a,y}$} \label{appendix: general efficiency theory psi_ay}

We provide the generalization of Lemma~\ref{lemma: remainder term for psi_01} for any $(a,y) \in \{0,1\}^2$.
%%%%%%%%%%%%%%%%%%%%%%%%% PSI efficiency %%%%%%%%%%%%%%%%%%%%%%%%
\begin{lemma}
\label{lemma: remainder term for psi_ay}

We have the following von Mises expansion for $\psi_{a,y}$
\begin{align*}
\psi_{a,y}(Q) &= \psi_{a,y}(\bar Q) + \int \varphi_{a,y}(\bar Q) d(Q-\bar Q) + R_2(\bar{Q}, Q) \\
 R_2(\bar{Q}, Q) &=  \frac{\bar \omega_y - \omega_y}{\bar \omega_y} \big(\psi_{a,y}(Q) - \psi_{a,y}(\bar Q)\big) + \frac{1}{\bar \omega_y}\int  \Big( \frac{\mu_a(x) - \bar \mu_a(x)}{\bar \mu_a(x)(1-\bar \mu_a(x))}\big( \eta_y(x) - \bar \eta_y(x) \big)  \\
    &+ \bar \eta_y(x) \frac{\mu_a(x) - \bar \mu_a(x)}{\bar \mu_a(x)(1-\bar \mu_a(x))} \frac{\bar \pi_a(x) - \pi_a(x)}{\bar \pi_a (x)} \\
    &+ \eta(x) \frac{\mu^{*}_a(x) -1/2}{\mu^{*}_a(x)^2(1-\mu^{*}_a(x))^2}\big(\mu_a(x) - \bar \mu_a(x)\big)^2\Big)dQ
\end{align*}

where $\mu_a^{*}(x)$ lies between $\bar{\mu}_a(x)$ and $\mu_a(x)$  and 
\begin{align*}
  \varphi_{a,y}(Z) &= \frac{\eta_y(X)\mathrm{logit}(\mu_a(X))}{\omega_y} - \psi_{a,y}  \\& +\frac{\eta_y(X)}{\omega_y} \frac{\Ind{A= a}(Y - \mu_a(X)}{\pi_a(X)\mu_a(X)(1-\mu_a(X))} + \frac{\mathrm{logit}(\mu_a(X))}{\omega_y}(\Ind{Y=y}- \eta_y(X)) \\
  &+ \psi_{a,y} - \frac{\psi_{a,y} \Ind{Y= y}}{\omega_y}
\end{align*}

Since the remainder term $R_2(\bar{Q}, Q)$ is a product of the nuisance function errors, we can apply Lemma 2 of \cite{kennedy2021semiparametric} to conclude that $\psi_{a,y}(Q)$ is pathwise differentiable with efficient influence function $\varphi_{a,y}(z; Q)$.
\end{lemma}

\subsubsection*{Proof of Lemma~\ref{lemma: remainder term for psi_ay}}
\begin{proof}
Let $f(x) = \mathrm{logit}(x)$. Then $f^{'}(x) = \frac{1}{x(1-x)}$ and $f^{''}(x) = \frac{2x-1}{x^2(1-x)^2}$. We apply Lemma~\ref{lemma: second order result for functions of regression functions} for $\psi_{a,y}$ to get 
\begin{align}
\begin{split}
\psi_{a,y}(Q) &= \psi_
{a,y}(\bar Q) + \int \varphi_{a,y}(\bar Q) d(Q-\bar Q) + R_2(\bar{Q}, Q)\\ 
 R_2(\bar{Q}, Q) &=  \frac{\bar \omega_y - \omega_y}{\bar \omega_y} \big(\psi_{a,y}(Q) - \psi_{a,y}(\bar Q)\big) + \frac{1}{\bar \omega_y}\int  \Big( \frac{\mu_a(x) - \bar \mu_a(x)}{\bar \mu_a(x)(1-\bar \mu_a(x))}\big( \eta_y(x) - \bar \eta_y(x) \big)  \\
    &+ \bar \eta_y(x) \frac{\mu_a(x) - \bar \mu_a(x)}{\bar \mu_a(x)(1-\bar \mu_a(x))} \frac{\bar \pi_a(x) - \pi_a(x)}{\bar \pi_a (x)} \\
    &+ \eta(x) \frac{\mu^{*}_a(x) -1/2}{\mu^{*}_a(x)^2(1-\mu^{*}_a(x))^2}\big(\mu_a(x) - \bar \mu_a(x)\big)^2\Big)dQ
\end{split}
\end{align}

where $\mu_a^{*}(x)$ is between $\bar{\mu}_a(x)$ and $\mu_a(x)$ and $\varphi_{a,y}$ is 
\begin{align*}
   \varphi_{a,y}(z;Q) &= \frac{\logit(\mu_a(X))}{\omega_y}\ind\{Y = y\}  +\frac{\eta_y(X)}{\omega_y } \frac{\ind\{A =a\}(Y - \mu_a(X))}{\pi_a(X)\mu_a(X) \big(1- \mu_a(X) \big)} - \frac{\psi_{a,y} \ind\{Y = y\}}{\omega_y}
\end{align*}

\end{proof}

\subsection{Proof of Theorem~\ref{theorem: efficiency bound for variance of influence function for geometric mean}} \label{appendix: proof of efficiency bound}
\begin{proof}
We first state some preliminaries that will be useful in proving the result.

First,  since $A$ is a binary random variable, then for any function $f(x,y)$,
\begin{align} \label{equation: expectation of two different values of binary rv is zero} 
    % \E[A(1-A) \mid X, Y] &=0 \\
    \E[\Ind{A=1}\Ind{A=0}f(X,Y)] = 0
\end{align}

Second, 
\begin{align}
    \E[\Ind{A=a}(Y -\mu_a(X))^2 \mid X] & = \mu_a(X)(1-\mu_a(X))\pi_a(X) \label{equation: A=a mu_a residual squared} \\
    \E[Y\Ind{A=a}(Y -\mu_a(X)) \mid X] &=  \mu_a(X)(1-\mu_a(X)) \pi_a(X)  \label{equation: Y times A mu_a residual}  \\% \var(Y \mid X, A =a) \pi_a(X) \\ 
    \E[(1-Y)\Ind{A=a}(Y -\mu_a(X)) \mid X] &=  -\mu_a(X)(1-\mu_a(X)) \pi_a(X) \label{equation: 1-Y times A mu_a residual} 
\end{align}

% For reference
% \begin{align*}
%   \varphi_{a,y}(Z) &= \frac{\eta_y(X)\mathrm{logit}(\mu_a(X))}{\omega_y} - \psi_{a,y}  \\& +\frac{\eta_y(X)}{\omega_y} \frac{\Ind{A= a}(Y - \mu_a(X)}{\pi_a(X)\mu_a(X)(1-\mu_a(X))} + \frac{\mathrm{logit}(\mu_a(X))}{\omega_y}(\Ind{Y=y}- \eta_y(X)) \\
%   &+ \psi_{a,y} - \frac{\psi_{a,y} \Ind{Y= y}}{\omega_y} \\
% %  &\varphi_{a,y}(z;Q) = \frac{\logit(\mu_a(X))}{\omega_y}\ind\{Y = y\}  +\frac{\eta_y(X)}{\omega_y } \frac{\ind\{A =a\}(Y - \mu_a(X))}{\pi_a(X)\mu_a(X) \big(1- \mu_a(X) \big)} - \frac{\psi_{a,y} \ind\{Y = y\}}{\omega_y}  
% \end{align*}

Using the influence function $\IF(\log(\gamma))  = \sum_{y=0}^1 \rho_y (\varphi_{1,y}(Z)-\varphi_{0,y}(Z))$, we can write the variance as 

\begin{align*}
    &\var\Big(\sum_{y=0}^1  \rho_y (\varphi_{1,y}(Z)-\varphi_{0,y}(Z)) \Big) = \\
    &\var\Bigg(\underbrace{\sum_{y=0}^1  \rho_y \Big(\frac{\Ind{Y= y}~\mathrm{logit}(\mu_1(X))}{\omega_y} -\frac{\Ind{Y= y}~\mathrm{logit}(\mu_0(X))}{\omega_y} \Big) - \sum_{y=0}^1  \rho_y \Big(\psi_{1,y}-\psi_{0,y} \Big)}_1 \\
    &+\underbrace{\sum_{y=0}^1  \rho_y \Big(\frac{\eta_y(X)}{\omega_y} \frac{\Ind{A= 1}(Y - \mu_1(X)}{\pi_1(X)\mu_1(X)(1-\mu_1(X))} - \frac{\eta_y(X)}{\omega_y} \frac{\Ind{A= 0}(Y - \mu_0(X)}{\pi_0(X)\mu_0(X)(1-\mu_0(X))} \Big)}_2 \\
    &+  \underbrace{\sum_{y=0}^1  \rho_y \Big( \psi_{1,y} - \psi_{0,y} \Big)  -  \sum_{y=0}^1  \rho_y \Big(\frac{\psi_{1,y} \Ind{Y= y}}{\omega_y}  - \frac{\psi_{0,y} \Ind{Y= y}}{\omega_y}\Big)}_3  \Bigg) = \\
      &\var\Bigg( \underbrace{\Psi(X,Y) - \log(\gamma)}_1  +  \underbrace{\log(\gamma) - \Psi^{*}(Y)}_3 \\
       &+ \underbrace{\Big(\rho_0\frac{\eta_0(X)}{\omega_0} + \rho_1\frac{\eta_1(X)}{\omega_1} \Big) \Big(\frac{\Ind{A= 1}(Y - \mu_1(X))}{\pi_1(X)\mu_1(X)(1-\mu_1(X))} - \frac{\Ind{A= 0}(Y - \mu_0(X))}{\pi_0(X)\mu_0(X)(1-\mu_0(X))} \Big)}_2  \Bigg) %= \\
    %   &\var \big(\Psi(X,Y)\big)  +  \var \big(\Psi^{*}(Y)\big) - 2 \covar\Big(\Psi(X,Y), \Psi^{*}(Y) \Big) + \\
    % &+ \var\Big(\rho_0\frac{\eta_0(X)}{\omega_0} + \rho_1\frac{\eta_1(X)}{\omega_1} \Big) \Big(\frac{\Ind{A= 1}(Y - \mu_1(X))}{\pi_1(X)\mu_1(X)(1-\mu_1(X))} - \frac{\Ind{A= 0}(Y - \mu_0(X))}{\pi_0(X)\mu_0(X)(1-\mu_0(X))} \Big) + \\
    %  &+ 2\covar\Bigg(\Psi(X,Y) - \Psi^{*}(Y), \Big(\rho_0\frac{\eta_0(X)}{\omega_0} + \rho_1\frac{\eta_1(X)}{\omega_1} \Big) \Big(\frac{\Ind{A= 1}(Y - \mu_1(X))}{\pi_1(X)\mu_1(X)(1-\mu_1(X))} - \frac{\Ind{A= 0}(Y - \mu_0(X))}{\pi_0(X)\mu_0(X)(1-\mu_0(X))} \Big)  \Bigg)
\end{align*}
where the first equality applies Eq.~\ref{equation: influence function of psi a y}, and the second applies Eq.~\ref{equation: definition of big PSI} and~\ref{equation: definition of big PSI STAR}. %equality uses the fact that the influence functions have mean zero.

Now let's consider the variance of term 2. Observe that term 2 has mean zero, so the variance of term 2 is: 
\begin{align*}
& \E \Bigg[\Big(\rho_0\frac{\eta_0(X)}{\omega_0} + \rho_1\frac{\eta_1(X)}{\omega_1} \Big)^2 \Big(\frac{\Ind{A= 1}(Y - \mu_1(X)}{\pi_1(X)\mu_1(X)(1-\mu_1(X))} - \frac{\Ind{A= 0}(Y - \mu_0(X)}{\pi_0(X)\mu_0(X)(1-\mu_0(X))} \Big) \Big)^2 \Bigg]= \\
 & \E \Bigg[ \Big(\rho_0\frac{\eta_0(X)}{\omega_0} + \rho_1\frac{\eta_1(X)}{\omega_1} \Big)^2  \Big( \frac{1}{\pi_1(X)\mu_1(X)(1-\mu_1(X))} +  \frac{1}{\pi_0(X)\mu_0(X)(1-\mu_0(X))} \Big) \Bigg]
\end{align*}

where we used iterated expectation, the fact that $\Ind{A=a}^2 = \Ind{A=a}$, Eq.~\ref{equation: A=a mu_a residual squared}, and Eq.~\ref{equation: expectation of two different values of binary rv is zero}.
We now convert to the notation scheme used in the main paper where $p = \rho_1$ and $\eta(x) = \eta_1(x)$ to rewrite 
\begin{align*}
    \Big(\rho_0\frac{\eta_0(X)}{\omega_0} + \rho_1\frac{\eta_1(X)}{\omega_1} \Big)  %\frac{\rho_0}{\omega_0}\Big(1-\eta_1(X)\Big) + \frac{\rho_1}{\omega_1}\eta_1(X) \\
    %&= \frac{\rho_0}{\omega_0} + \Bigg(\frac{\rho_1}{\omega_1} - \frac{\rho_0}{\omega_0} \Bigg)\eta_1(X) \\
    % &= \frac{1-p}{1-q} + \Bigg(\frac{p}{q} - \frac{1-p}{1-q} \Bigg)\Bigg(\pi_0(X)\mu_0(X) + \pi_1(X)\mu_1(X)\Bigg)\\
    % &= \frac{1-p}{1-q} + \Bigg(\frac{p(1-q) -(1-p)q}{q(1-q)} \Bigg)\Bigg(\pi_0(X)\mu_0(X) + \pi_1(X)\mu_1(X)\Bigg)\\
    &= \frac{1-\rho}{1-\omega} + \Bigg(\frac{\rho - \omega}{\omega(1-\omega)} \Bigg)\eta(X)\\
\end{align*}

% to write the expectation of square of term 2 as
% \begin{align*} 
%  & \E \Bigg[ \Bigg(\frac{1-p}{1-q} + \Bigg(\frac{p -q}{q(1-q)} \Bigg)\eta_1(X) \Bigg)^2  \Big( \frac{1}{\pi_1(X)\mu_1(X)(1-\mu_1(X))} +  \frac{1}{\pi_0(X)\mu_0(X)(1-\mu_0(X))} \Big) \Bigg]
% \end{align*}

Now let's consider the covariance: 
\begin{align*}
    &\covar\Bigg(\Psi(X,Y) - \Psi^{*}(Y), \Big(\rho_0\frac{\eta_0(X)}{\omega_0} + \rho_1\frac{\eta_1(X)}{\omega_1} \Big) \Big(\frac{\Ind{A= 1}(Y - \mu_1(X))}{\pi_1(X)\mu_1(X)(1-\mu_1(X))} - \frac{\Ind{A= 0}(Y - \mu_0(X))}{\pi_0(X)\mu_0(X)(1-\mu_0(X))} \Big)  \Bigg) = \\
    &\E\Bigg[\Big(\Psi(X,Y) - \Psi^{*}(Y) \Big) \Big(\rho_0\frac{\eta_0(X)}{\omega_0} + \rho_1\frac{\eta_1(X)}{\omega_1} \Big) \Big(\frac{\Ind{A= 1}(Y - \mu_1(X))}{\pi_1(X)\mu_1(X)(1-\mu_1(X))} - \frac{\Ind{A= 0}(Y - \mu_0(X))}{\pi_0(X)\mu_0(X)(1-\mu_0(X))} \Big)  \Bigg] = \\
    &\E\Bigg[ \Big(\rho_0\frac{\eta_0(X)}{\omega_0} + \rho_1\frac{\eta_1(X)}{\omega_1} \Big) \E \Big[ \Big(\Psi(X,Y) - \Psi^{*}(Y) \Big)  \Big(\frac{\Ind{A= 1}(Y - \mu_1(X))}{\pi_1(X)\mu_1(X)(1-\mu_1(X))} - \frac{\Ind{A= 0}(Y - \mu_0(X))}{\pi_0(X)\mu_0(X)(1-\mu_0(X))} \Big) \mid X \Big]  \Bigg] \\
    &= 0
\end{align*}
where the first equality made use of the fact that both terms are mean zero, the second  applies iterated expectation, and the last line used Eqs.~\ref{equation: Y times A mu_a residual}-\ref{equation: 1-Y times A mu_a residual}.

Putting this all together gives 
\begin{align*}
          &\var(\IF(\log(\gamma))) = \var \big(\Psi(X,Y)\big)  +  \var \big(\Psi^{*}(Y)\big) - 2 \covar\Big(\Psi(X,Y), \Psi^{*}(Y) \Big) \\
    &+ \E \Bigg[ \Bigg(\frac{1-\rho}{1-\omega} + \Bigg(\frac{\rho -\omega}{\omega(1-\omega)} \Bigg)\eta(X) \Bigg)^2  \Big( \frac{1}{\pi_1(X)\mu_1(X)(1-\mu_1(X))} +  \frac{1}{\pi_0(X)\mu_0(X)(1-\mu_0(X))} \Big) \Bigg]
\end{align*}

\end{proof}

% \subsubsection{Equivalence of definitions~\ref{definition: geometric odds ratio} of geometric odds ratio}
% The equivalence between Eq.s~\ref{equation: geometric odds ratio def1} and \ref{equation: geometric odds ratio def2} holds due to the properties of exponents and multiplication. Observe that from Eq.~\ref{equation: geometric odds ratio def1} we have 
% \begin{align*}
%     \log \gamma &= \int \log \Bigg( \frac{\frac{P(Y^1 =1 \mid X = x)}{P(Y^1=0\mid X = x)}}{\frac{P(Y^0 =1\mid X = x)}{P(Y^0=0\mid X = x)}}\Bigg) dP(x) 
% \end{align*}
% Exponentiating both sides yields~\ref{equation: geometric odds ratio def3}.

\subsection{Additional results: Random sampling} \label{section: efficiency under random sampling}

In this section we consider the simpler setting of estimating $\gamma$ under random sampling. Although causal estimands like the risk difference are more commonly used in random sampling, the odds ratio may nonetheless be of interest. Additionally, this analysis will differentiate which of the factors contributing to the difficulty of estimation (see Section~\ref{section: efficiency}) originate from the choice of estimand versus which result from the sampling scheme.

Under random sampling, we can point identify $\gamma$ (see Eq.~\ref{equation: identification of GOR under random sampling}).

\begin{lemma} \label{lemma: remainder term for target under random sampling}
We have the following von Mises expansion for the target geometric mean:

% \begin{align}
%     \exp(\psi(P)) = \exp(\psi(\bar{P})) + \exp(\psi(\bar{P}))\int \phi(\bar{P})d(P- \bar{P})) + R_2(\bar{P}, P)
% \end{align}

\begin{align}
    \gamma(P) = \gamma(\bar{P}) + \gamma(\bar{P})\int \phi(\bar{P})d(P- \bar{P})) + R_2(\bar{P}, P)
\end{align}
with 
\begin{align*}
R_2(\bar{P}, P) &= \int \frac{(\bar{\pi}(x) - \pi(x))(\mu_1(x) -\bar{\mu}_1(x))}{\bar{\mu}_1(x)(1-\bar{\mu}_1(x))\bar{\pi}(x)} + \frac{(\pibar(x) - \pi(x))(\mu_0(x) -\bar{\mu}_0(x))}{\bar{\mu}_0(x)(1-\bar{\mu}_0(x))(1-\pibar(x))} \\
 &+ \frac{(\mu_1^{*}(x)-1/2)(\mu_1(x) -\bar{\mu}_1(x))^2}{\mu_1^{*}(x)^2(1-\mu_1^{*}(x))^2} - \frac{(\mu_0^{*}(x)-1/2)(\mu_0(x) -\bar{\mu}_0(x))^2}{\mu_0^{*}(x)^2(1-\mu_0^{*}(x))^2}dP(x)
\end{align*}

where $\mu_a^{*}(x)$ lies between $\bar{\mu}_a(x)$ and $\mu_a(x)$ for $a \in \{0,1\}$ and 
\begin{align*}
    \varphi(z; P) = \log\Bigg(\frac{\frac{\mu_1(X)}{1-\mu_1(X)}}{\frac{\mu_0(X)}{1-\mu_0(X)}}\Bigg) - \log(\gamma) + \frac{A(Y - \mu_1(X))}{\mu_1(X)(1-\mu_1(X))\pi(X)}  -  \frac{(1-A)(Y - \mu_0(X))}{\mu_0(X)(1-\mu_0(X))(1-\pi(X))}.
\end{align*}

\end{lemma}

\begin{proof} \label{proof of corollary: remainder term for target under random sampling}
For a $\gamma^{*}$ such that $\log(\gamma^{*})$ lies between $\log(\gamma(\bar{P}))$ and $\log(\gamma(P))$, applying Taylor's Theorem  
\begin{align}
    \gamma(P) &= \gamma(\bar{P}) + \gamma(\bar{P})(\log(\gamma(P)) - \log(\gamma(\bar{P}))) + \frac{1}{2} (\log(\gamma(P)) - \log(\gamma(\bar{P})))^2\gamma^{*}\\
    &= \gamma(\bar{P}) + \gamma(\bar{P})\int \varphi(\bar{P})d(P- \bar{P})) + R_2(\bar{P}, P)
\end{align} 

where the second equality applies Lemma~\ref{lemma: remainder term under random sampling}.

Then, by Lemma 2 of \citep{kennedy2021semiparametric}, our target $\gamma(P)$ is pathwise differentiable with influence function $\gamma(P)\phi(z;P)$.
\end{proof}

\begin{theorem} \label{theorem: efficiency bound for gamma under random sampling} The nonparametric efficiency bound for estimating $\gamma$ under random sampling is given by $\sigma^2 := \gamma^2 \var(\phi(Z))$ where
\begin{align*}
    % \var(\phi(Z)) &= \E \Bigg[ \frac{1}{\pi(X)\var(Y \mid X, A=1)} + \frac{1}{(1-\pi(X))\var(Y \mid X, A=0)} + \Bigg(\log\Big(\frac{\frac{\muone(X)}{1-\muone(X)}}{\frac{\muz(X)}{1-\muz(X)}} \Big) - \psi \Bigg)^2 \Bigg] \\
        \var(\phi(Z)) &= \E \Bigg[ \frac{1}{\pi(X)\muone(X)(1-\muone(X))} + \frac{1}{(1-\pi(X))\muz(X)(1-\muz(X))} + \Bigg(\log\Big(\frac{\frac{\muone(X)}{1-\muone(X)}}{\frac{\muz(X)}{1-\muz(X)}} \Big) - \log(\gamma) \Bigg)^2 \Bigg]
\end{align*}
\end{theorem}

Theorem~\ref{theorem: efficiency bound for gamma under random sampling} indicates that the difficulty of estimation problem under random sampling depends only on three factors:
\begin{enumerate}
    \item Propensity scores $\pi(x)$;
    \item Regression function variances $\muone(x)(1-\muone(x))$ and $\muz(x)(1-\muz(x))$; %(Equivalently conditional variances $\var(Y \mid X, A =a)$ for $a \in \{0,1\}$)
    \item Heterogeneity in the log odds ratio $\E\Bigg[\Bigg(\log\Big(\frac{\frac{\muone(X)}{1-\muone(X)}}{\frac{\muz(X)}{1-\muz(X)}} \Big) - \log(\gamma) \Bigg)^2\Bigg]$.
\end{enumerate}
The efficiency bound \emph{decreases} with the regression function variances, which is notably different from the average risk difference efficiency bound \cite{hahn1998role}.

\begin{proof}
First we will state some preliminaries that are useful in deriving the variance of the efficient influence function.

For any function $g(X)$,
\begin{align}
    \E[A(Y -\muone(X))g(X)] &=0 \label{equation: AY muone has zero expectation} \\
    \E[(1-A)(Y -\muz(X))g(X)] &=0 \label{equation: 1-A Y muzero has zero expectation}
\end{align}
This holds from iterated expectation:
\begin{align*}
    \E[A(Y -\muone(X))g(X)] &= \E[ \E[A(Y -\muone(X)) \mid X] g(X)] \\
    &= \E[\E[(Y -\muone(X)) \mid X, A= 1] \pi(X) g(X)]  = 0. 
\end{align*}

Additionally,
\begin{align}
    \E[A(Y -\muone(X))^2 \mid X] &= \E[(Y -\muone(X))^2 \mid X, A = 1]\pi(X) \label{equation: AY muone squared} \\
    \E[(1-A)(Y -\muz(X))^2 \mid X] &= \E[(Y -\muz(X))^2 \mid X, A = 0](1-\pi(X)) \label{equation: 1-A Y muzero squared}.
\end{align}

We also restate the influence function for $\log(\gamma)$:

\begin{align}\label{equation: influence function definition under random sampling}
    \varphi &= \log \Bigg( \frac{\frac{\mu_1(X)}{1-\mu_1(X)}}{\frac{\mu_0(X)}{1-\mu_0(X)}} \Bigg) +  \frac{A}{\pi(X)}  \frac{Y -\mu_1(x)}{(1-\mu_1(x))\mu_1(x)}   - \frac{1-A}{1-\pi(X)}  \frac{Y-\mu_0(x)}{(1-\mu_0(x)) \mu_0(x)} - \log(\gamma).
\end{align}

With these in hand, we have
\begin{align*}
    \var(\phi) &= \var(\varphi) = \E[\varphi^2] \\
    &= \E \Bigg[ \Bigg(\log\Big(\frac{\frac{\muone(X)}{1-\muone(X)}}{\frac{\muz(X)}{1-\muz(X)}} \Big) - \log(\gamma) \Bigg)^2  + \Bigg( \frac{A}{\pi(X)}  \frac{Y -\muone(X)}{(1-\muone(X))\muone(X)}   - \frac{1-A}{1-\pi(X)}  \frac{Y-\muz(X)}{(1-\muz(X)) \muz(X)}\Bigg)^2 \\
    &+ 2 \Bigg(\log\Big(\frac{\frac{\muone(X)}{1-\muone(X)}}{\frac{\muz(X)}{1-\muz(X)}} \Big) - \log(\gamma) \Bigg)\Bigg(\frac{A}{\pi(X)}  \frac{Y -\muone(X)}{(1-\muone(X))\muone(X)}   - \frac{1-A}{1-\pi(X)}  \frac{Y-\muz(X)}{(1-\muz(X)) \muz(x)} \Bigg)  \Bigg] \\
  &= \E \Bigg[ \Bigg(\log\Big(\frac{\frac{\muone(X)}{1-\muone(X)}}{\frac{\muz(X)}{1-\muz(X)}} \Big) - \log(\gamma) \Bigg)^2  + \Bigg( \frac{A}{\pi(X)}  \frac{Y -\muone(X)}{(1-\muone(X))\muone(X)}   - \frac{1-A}{1-\pi(X)}  \frac{Y-\muz(X)}{(1-\muz(X)) \muz(X)}\Bigg)^2  \Bigg] 
\end{align*}
where the first line uses the fact that the centered influence function has mean-zero, the second line uses our definition of $\varphi$ in Equation~\ref{equation: influence function definition under random sampling}, and the third uses the properties~\ref{equation: AY muone has zero expectation}-\ref{equation: 1-A Y muzero has zero expectation}.
We can simplify the second term:
\begin{align*}
    & \E \Bigg[ \Bigg( \frac{A}{\pi(X)}  \frac{Y -\muone(X)}{(1-\muone(X))\muone(X)}   - \frac{1-A}{1-\pi(X)}  \frac{Y -\muz(X)}{(1-\muz(X)) \muz(X)}\Bigg)^2  \Bigg] \\
    & =\E \Bigg[  \frac{A^2}{\pi(X)^2}  \frac{(Y -\muone(X))^2}{(1-\muone(X))^2\muone(X)^2}   - \frac{(1-A)^2}{(1-\pi(X))^2}  \frac{(Y -\muz(X))^2}{(1-\muz(X))^2 \muz(X)^2}  \\
    &-2 \frac{A}{\pi(X)}  \frac{Y -\muone(X)}{(1-\muone(X))\muone(X)} \frac{1-A}{1-\pi(X)}  \frac{Y -\muz(X)}{(1-\muz(X)) \muz(X)}  \Bigg]
    \\& =\E \Bigg[  \frac{A}{\pi(X)^2}  \frac{(Y -\muone(X))^2}{(1-\muone(X))^2\muone(X)^2}   - \frac{1-A}{(1-\pi(X))^2}  \frac{(Y -\muz(X))^2}{(1-\muz(X))^2 \muz(X)^2}  \Bigg]
    \\& =\E \Bigg[  \frac{1}{\pi(X)}  \frac{\E[(Y -\muone(X))^2 \mid X, A =1]}{(1-\muone(X))^2\muone(X)^2}   - \frac{1}{(1-\pi(X)} \frac{\E[(Y -\muz(X))^2\mid X, A = 0]}{(1-\muz(X))^2 \muz(X)^2}  \Bigg]
     \\& =\E \Bigg[  \frac{1}{\pi(X)(1-\muone(X))\muone(X)}   - \frac{1}{(1-\pi(X))(1-\muz(X)) \muz(X)}  \Bigg]
\end{align*}

The second equality uses the fact that $A^2 =A$ which also implies that $A(1-A)=0$. The third uses iterated expectation and properties~\ref{equation: AY muone squared}-\ref{equation: 1-A Y muzero squared}. The fourth uses the fact that $\E[(Y -\muone(X))^2 \mid X, A =1] = (1-\muone(X))\muone(X)$ and $\E[(Y -\muz(X))^2 \mid X, A =0] = (1-\muz(X))\muz(X)$.

\end{proof}

% \subsubsection{Efficiency and estimation of $\gamma$ under random sampling} \label{section: efficiency under random sampling}
 We propose the estimator 

\begin{align*}
\hat{\gamma}_P = \exp\big(\pn(&\phi(Z; \hat{\mu}_1, \hat{\mu}_0, \hat{\pi}))\big) \quad \mathrm{where} \\
    &\phi(Z; \mu_1, \mu_0, \pi)) = \log \Bigg( \frac{\frac{\mu_1(X)}{1-\mu_1(X)}}{\frac{\mu_0(X)}{1-\mu_0(X)}} \Bigg) +  \frac{A}{\pi(X)}  \frac{Y -\mu_1(x)}{(1-\mu_1(x))\mu_1(x)}   - \frac{1-A}{1-\pi(X)}  \frac{Y-\mu_0(x)}{(1-\mu_0(x)) \mu_0(x)}.
\end{align*}

\begin{corollary}\label{corollary: asymptotic normality of estimator under random sampling}
The estimator $\hat{\gamma}_P$ is $\sqrt{n}$-consistent and asymptotically normal 
under the identification assumptions (\ref{assumption:consistency}-\ref{assumption:positivity})
and the following conditions: 
\begin{enumerate}
    \item \emph{Strong overlap:} $\p(\epsilon < \pi(X) < 1- \epsilon) = 1$  and $\p(\epsilon < \pih(X) < 1- \epsilon) = 1$ for some $\epsilon \in (0,1)$. 
    \item \emph{Outcome variance:} $\p(\mu_a(X)(1-\mu_a(X)) > \epsilon) = 1$ and $\p(\hat{\mu}_a(X)(1-\hat{\mu}_a(X)) > \epsilon) = 1$ for $a \in \{0,1\}$ and for some $\epsilon \in (0,1)$. 
    \item \emph{Convergence in probability in $L_2(\p)$ norm:} $\norm{\phi - \phih} = o_{\p}(1)$. This can be achieved by using plug-in estimates of the nuisance functions. %the consistency of the empirical process term
    \item \emph{Sample-splitting:} $\hat{\pi}$, $\hat{\mu}_0$ and $\hat{\mu}_1$ are estimated on samples from $\hat P$.
    \item $\norm{\hat{\pi} -\pi} = O_{\p}(n^{-1/4})$.
    \item $\norm{\hat{\mu}_a -\mu_a} = o_{\p}(n^{-1/4})$ for $a \in \{0,1\}$.
\end{enumerate}

The limiting distribution is $\sqrt{n}(\hat{\gamma}_P - \gamma) \rightsquigarrow \mathcal{N}(0, \sigma^2)$ where $\sigma^2$ is given in the next theorem.
\end{corollary}
\begin{proof}
Lemma~\ref{lemma: Error of psi estimator under random sampling}.
\end{proof}

% Therefore we propose the estimator $\exp\big(\pn(\varphi(Z; \hat{\mu}_1, \hat{\mu}_0, \hat{\pi}))\big)$ % this is what we get by solving the estimating equations for the exp(psi)

\subsubsection{Useful lemmas for the random sampling setting}

% shortcut to remainder
\begin{lemma}
\label{lemma: remainder term under random sampling}
We have the following von Mises expansion of $\log(\gamma)$:
\begin{align*}
\log(\gamma(P)) = \log(&\gamma(\bar P)) + \int \varphi(\bar P) d(P-\bar P) + R_2(\bar{P}, P) \\
\mathrm{for} \quad   R_2(\bar{P}, P) = &\int \frac{(\bar{\pi}(x) - \pi(x))(\mu_1(x) -\bar{\mu}_1(x))}{(\bar{\mu}_1(x))(1-\bar{\mu}_1(x))\bar{\pi}(x)} + \frac{(\pibar(x) - \pi(x))(\mu_0(x) -\bar{\mu}_0(x))}{(\bar{\mu}_0(x))(1-\bar{\mu}_0(x))(1-\pibar(x))} \\
 &+ \frac{(\mu_1^{*}(x)-1/2)(\mu_1(x) -\bar{\mu}_1(x))^2}{\mu_1^{*}(x)^2(1-\mu_1^{*}(x))^2} - \frac{(\mu_0^{*}(x)-1/2)(\mu_0(x) -\bar{\mu}_0(x))^2}{\mu_0^{*}(x)^2(1-\mu_0^{*}(x))^2}dP(x)
\end{align*}

% \begin{align*}
% \psi(P) &= \psi(\bar P) + \int \varphi(\bar P) d(P-\bar P) + R_2(\bar{P}, P) \\
% R_2(\bar{P}, P) &= \int \frac{(\bar{\pi}(x) - \pi(x))(\mu_1(x) -\bar{\mu}_1(x))}{(\bar{\mu}_1(x))(1-\bar{\mu}_1(x))\bar{\pi}(x)} + \frac{(\pibar(x) - \pi(x))(\mu_0(x) -\bar{\mu}_0(x))}{(\bar{\mu}_0(x))(1-\bar{\mu}_0(x))(1-\pibar(x))} \\
%  &+ \frac{(\mu_1^{*}(x)-1/2)(\mu_1(x) -\bar{\mu}_1(x))^2}{\mu_1^{*}(x)^2(1-\mu_1^{*}(x))^2} - \frac{(\mu_0^{*}(x)-1/2)(\mu_0(x) -\bar{\mu}_0(x))^2}{\mu_0^{*}(x)^2(1-\mu_0^{*}(x))^2}dP(x)
% \end{align*}

where $\mu_a^{*}(x)$ lies between $\bar{\mu}_a(x)$ and $\mu_a(x)$ for $a \in \{0,1\}$ and 
\begin{align*}
    \varphi(z; P) = \log\Bigg(\frac{\frac{\mu_1(X)}{1-\mu_1(X)}}{\frac{\mu_0(X)}{1-\mu_0(X)}}\Bigg) - \log(\gamma) + \frac{A(Y - \mu_1(X))}{\mu_1(X)(1-\mu_1(X))\pi(X)}  -  \frac{(1-A)(Y - \mu_0(X))}{\mu_0(X)(1-\mu_0(X))(1-\pi(X))}.
\end{align*}

%Since the remainder term $R_2(\bar{P}, P)$ is a product of the nuisance function errors, we can apply Lemma 2 of \cite{kennedy2021semiparametric} to conclude that $\log(\gamma(P))$ is pathwise differentiable with efficient influence function $\varphi(z; P)$.
\end{lemma}

\subsubsection*{Proof of Lemma~\ref{lemma: remainder term under random sampling}}
\begin{proof}
We can write the log geometric mean as 
\begin{align*} 
    \log(\gamma(P)) &=  \psi_1(P) -  \psi_0(P) \quad  \mathrm{where} \\
   \psi_1(P) &= \int_X  \log \Big(\frac{\mu_1(x)}{1-\mu_1(x)}  \Big) dP(x) \\
   \psi_0(P) &= \int_X  \log \Big(\frac{\mu_0(x)}{1-\mu_0(x)}  \Big) dP(x) 
\end{align*}

Then, let $f(x) = \log\Big(\frac{x}{1-x}\Big)$ and apply Lemma~\ref{lemma: second order result for functions of regression functions} for $\psi_a$ to get 
\begin{align*}
\begin{split}
\psi_a(P) &= \psi_a(\bar P) + \int \varphi_a(\bar P) d(P-\bar P) + R_2(\bar{P}, P)\\ 
 R_2(\bar{P}, P) &= \int \frac{(\bar{\pi}_a(x) - \pi_a(x))(\mu_a(x) -\bar{\mu}_a(x))}{(\bar{\mu}_a(x))(1-\bar{\mu}_a(x))\bar{\pi}_a(x)} + \frac{(\mu_a^{*}(x)-1/2)(\mu_a(x) -\bar{\mu}_a(x))^2}{\mu_a^{*}(x)^2(1-\mu_a^{*}(x))^2}dP(x).
\end{split}
\end{align*}

Then we subtract the von Mises expansion for $\psi_0(P)$ from the expansion for $\psi_1(P)$  to obtain
\begin{align*}
\begin{split}
\log(\gamma(P)) &= \log(\gamma(\bar P)) + \int \varphi_1(\bar P) - \varphi_0(\bar P) d(P-\bar P) + R_2(\bar{P}, P)\\ 
 R_2(\bar{P}, P) &= \int \frac{(\bar{\pi}(x) - \pi(x))(\mu_1(x) -\bar{\mu}_1(x))}{(\bar{\mu}_1(x))(1-\bar{\mu}_1(x))\bar{\pi}(x)} + \frac{(\pibar(x) - \pi(x))(\mu_0(x) -\bar{\mu}_0(x))}{(\bar{\mu}_0(x))(1-\bar{\mu}_0(x))(1-\pibar(x))} \\
 &+ \frac{(\mu_1^{*}(x)-1/2)(\mu_1(x) -\bar{\mu}_1(x))^2}{\mu_1^{*}(x)^2(1-\mu_1^{*}(x))^2} - \frac{(\mu_0^{*}(x)-1/2)(\mu_0(x) -\bar{\mu}_0(x))^2}{\mu_0^{*}(x)^2(1-\mu_0^{*}(x))^2}dP(x).
\end{split}
\end{align*}
\end{proof}

% estimator for log geometric mean under random sampling
We propose the following to estimate $\log(\gamma)$: $\hat \psi := \pn(\phi(Z; \hat{\mu}_1, \hat{\mu}_0, \hat{\pi}))$ where
\begin{align*} %\label{equation: proposed log estimator under random sampling}
\begin{split}
\phi(Z; \mu_1, \mu_0, \pi) = \log\Bigg(\frac{\frac{\mu_1(X)}{1-\mu_1(X)}}{\frac{\mu_0(X)}{1-\mu_0(X)}}\Bigg) + \frac{A(Y - \mu_1(X))}{\mu_1(X)(1-\mu_1(X))\pi(X)}  -  \frac{(1-A)(Y - \mu_0(X))}{\mu_0(X)(1-\mu_0(X))(1-\pi(X))}.
\end{split}
\end{align*}

\begin{lemma} \label{lemma: Error of psi estimator under random sampling}
Assume the identification assumptions (\ref{assumption:consistency}-\ref{assumption:positivity}) hold 
and additionally assume the following conditions hold: 
\begin{enumerate}
    \item \emph{Strong overlap:} $\p(\epsilon < \pi(X) < 1- \epsilon) = 1$  and $\p(\epsilon < \pih(X) < 1- \epsilon) = 1$ for some $\epsilon \in (0,1)$. 
    \item \emph{Outcome variance:} $\p(\mu_a(X)(1-\mu_a(X)) > \epsilon) = 1$ and $\p(\hat{\mu}_a(X)(1-\hat{\mu}_a(X)) > \epsilon) = 1$ for $a \in \{0,1\}$ and for some $\epsilon \in (0,1)$. 
    \item \emph{Convergence in probability in $L_2(\p)$ norm:} $\norm{\phi - \phih} = o_{\p}(1)$. This can be achieved by using plug-in estimates of the nuisance functions. %the consistency of the empirical process term
    \item \emph{Sample-splitting:} $\hat{\pi}$, $\hat{\mu}_0$ and $\hat{\mu}_1$ are estimated on samples from $\hat P$.
\end{enumerate}
For the error of the proposed estimator, we have $ \pn(\phi(Z; \hat{\mu}_1, \hat{\mu}_0, \hat{\pi})) - \log(\gamma) = $
\begin{align*}
 (\pn - \p) (\varphi(Z; \p))  + O_{\p}\Bigg(\sum_{a=0}^{1}\Big( \norm{\hat{\pi} - \pi}\norm{\mu_a -\hat{\mu}_a} + \norm{\mu_a -\hat{\mu}_a}^2 \Big) \Bigg) + o_{\p}\Big(\frac{1}{\sqrt{n}}\Big)
\end{align*}
\end{lemma}

\begin{proof}
We define the shorthand $\psi = \log(\gamma)$ and $\hat{\psi} =  \pn(\phi(Z; \hat{\mu}_1, \hat{\mu}_0, \hat{\pi}))$.
\begin{align*}
\hat{\psi} - \psi &= \pn(\hat \phi) - \psi(\p) \\
&= \pn(\hat \varphi) + \psi(\hat \p) - \psi(\p) \\
&= (\pn-\p)(\hat \varphi - \varphi)  + \p(\hat \varphi - \varphi) + \pn(\varphi) + \psi(\hat \p) -  \psi(\p) \\
&= (\pn-\p)(\hat \varphi - \varphi)  + (\pn - \p)(\varphi) +  \p(\hat \varphi) + \psi(\hat \p) - \psi(\p) \\
&= (\pn-\p)(\hat \varphi - \varphi)  + (\pn - \p)(\varphi) +  R_2(\hat \p, \p)
\end{align*}

where $R_2(\hat \p, \p)$ is given in Lemma~\ref{lemma: remainder term under random sampling}.

Under the conditions of this lemma we have 
\begin{align*}
     R_2(\hat \p, \p) = O_P\Bigg(\int &(\hat \pi (x) - \pi(x))(\hat \mu_1 (x) - \mu_1(x))dP(x) + (\hat \pi (x) - \pi(x))(\hat \mu_0 (x) - \mu_0(x)) \\&+ (\hat \mu_1 (x) - \mu_1(x))^2 + 
     (\hat \mu_0 (x) - \mu_0(x))^2dP(x) \Bigg).
\end{align*}

Applying the Cauchy-Schwarz inequality yields
\begin{align*}
    R_2(\hat \p, \p) = O_P\Bigg(&\norm{\hat \pi (x) - \pi(x)}\norm{\hat \mu_1 (x) - \mu_1(x)} + \norm{\hat \pi (x) - \pi(x)}\norm{\hat \mu_0 (x) - \mu_0(x)} \\&+ \norm{\hat \mu_1 (x) - \mu_1(x)}^2 + 
     \norm{\hat \mu_0 (x) - \mu_0(x)}^2 \Bigg) .
\end{align*}

Because $P_n$ is the empirical measure on an independent sample from $\hat \p$, we can apply Lemma 2 of \citet{kennedy2020sharp} along with Condition 3 of the Lemma statement to obtain
\begin{align*}
     (\pn-\p)(\hat \varphi - \varphi) = o_p(\frac{1}{\sqrt{n}}).
\end{align*}

Putting this all together gives
\begin{align*}
\hat \psi - \psi &=(\p- \pn) (\varphi(Z; \p)) + O_{\p}\Bigg(\sum_{a=0}^{1}\Big( \norm{\hat{\pi} - \pi}\norm{\mu_a -\hat{\mu}_a} + \norm{\mu_a -\hat{\mu}_a}^2 \Big) \Bigg) + o_{\p}\Big(\frac{1}{\sqrt{n}}\Big).
\end{align*}

\end{proof}

\clearpage
\bibliographystyle{plainnat}
\bibliography{ref}
\end{document}